\documentclass[journal]{IEEEtran}
\ifCLASSINFOpdf
\else
\fi
\usepackage[cmex10]{amsmath}
\usepackage{amsthm}
\usepackage{amsmath}
\usepackage{wrapfig}
\usepackage{nccmath}
\usepackage{amssymb}
\usepackage{graphicx}
\usepackage{cite}
\usepackage{cuted}

\newtheorem{definition}{Definition}[]
\newtheorem{theorem}{Theorem}[]
\newtheorem{prop}{Proposition}[]

\newtheorem{corollary}{Corollary}[]
\newtheorem{innercustomgeneric}{\customgenericname}
\providecommand{\customgenericname}{}
\newcommand{\newcustomtheorem}[2]{%
  \newenvironment{#1}[1]
  {%
   \renewcommand\customgenericname{#2}%
   \renewcommand\theinnercustomgeneric{##1}%
   \innercustomgeneric
  }
  {\endinnercustomgeneric}
}
\newcustomtheorem{customlemma}{Lemma}
\newcommand\smallO{
  \mathchoice
    {{\scriptstyle\mathcal{O}}}
    {{\scriptstyle\mathcal{O}}}
    {{\scriptscriptstyle\mathcal{O}}}
    {\scalebox{.6}{$\scriptscriptstyle\mathcal{O}$}}
  }

\usepackage{color}

\usepackage[hyperfootnotes=false]{hyperref}
\usepackage{cleveref}
\definecolor{darkblue}{rgb}{0,0,0.5}
\hypersetup{
colorlinks=true,
linkcolor=black,
filecolor=blue,
citecolor=darkblue,  
urlcolor=black,
}

\begin{document}
%
\title{Information contraction in noisy binary neural networks and its implications
}
%
%
%

\author{Chuteng~Zhou,
		Quntao~Zhuang,
        Matthew~Mattina,
        ~and~Paul~N.~Whatmough

\thanks{C. Zhou is with the Arm ML Research Lab, Boston, MA USA (e-mail: chu.zhou@arm.com).}
\thanks{Q. Zhuang is with the Department of Electrical and Computer Engineering and James C. Wyant College of Optical Sciences, University of Arizona, Tuscon, AZ USA (e-mail: zhuangquntao@email.arizona.edu).}%
\thanks{M. Mattina is with the Arm ML Research Lab, Boston, MA USA (e-mail: matthew.mattina@arm.com)}%
\thanks{P. N. Whatmough is with the Arm ML Research Lab, Boston, MA USA (e-mail: paul.whatmough@arm.com).}%
}

%
%

\markboth{Preprint, January~2021}%
{Shell \MakeLowercase{\textit{et al.}}: Bare Demo of IEEEtran.cls for Journals}
%



\maketitle

\begin{abstract}
Neural networks have gained importance as the machine learning models that achieve state-of-the-art performance on large-scale image classification, object detection and natural language processing tasks. In this paper, we consider noisy binary neural networks, where each neuron has a non-zero probability of producing an incorrect output. These noisy models may arise from biological, physical and electronic contexts and constitute an important class of models that are relevant to the physical world. 
Intuitively, the number of neurons in such systems has to grow to compensate for the noise while maintaining the same level of expressive power and computation reliability. 
Our key finding is a lower bound for the required number of neurons in noisy neural networks, which is first of its kind. To prove this lower bound, we take an information theoretic approach and obtain a novel strong data processing inequality (SDPI), which not only generalizes the Evans-Schulman results for binary symmetric channels to general channels, but also improves the tightness drastically when applied to estimate end-to-end information contraction in networks. Our SDPI can be applied to various information processing systems, including neural networks and cellular automata. 
Applying the SDPI in noisy binary neural networks, we obtain our key lower bound and investigate its implications on network depth-width trade-offs, our results suggest a depth-width trade-off for noisy neural networks that is very different from the established understanding regarding noiseless neural networks. 
Furthermore, we apply the SDPI to study fault-tolerant cellular automata and obtain bounds on the error correction overheads and the relaxation time. 
This paper offers new understanding of noisy information processing systems through the lens of information theory.
\end{abstract}


%

\IEEEpeerreviewmaketitle

\section{Introduction}
%
%
%
%

Inspired by biological neural systems, artificial neural networks have gained great success in a wide variety of tasks such as computer vision, speech recognition and natural language processing. The majority of modern deep neural networks are based on artificial neurons which have a deterministic behavior. There has been rising interest in the study of neural networks with stochastic neurons. These noisy neural networks can be encountered in physical implementations of neural networks with analog components which are intrinsically noisy \cite{lin2018all, yao2020fully, shen2017deep}. Studies of animal brains show that biological neurons are stochastic and our perception is fundamentally limited by the information that the brain can extract from the noisy dynamics of sensory neurons \cite{rumyantsev2020fundamental}. Furthermore, neurons that fire stochastically are used in models such as the restricted Boltzmann machines \cite{fischer2014training} and are shown to enhance the robustness in some deep neural networks \cite{bengio2013estimating}. Information theory provides a rigorous and insightful tool to model and understand noise in these systems.

Von Neumann \cite{von1956probabilistic} investigated the feasibility of computing Boolean functions with a circuit composed of noisy gates with a bounded number of inputs. He arrived at the conclusion that a Boolean function can be $\delta$-reliably (meaning the probability of producing a wrong answer for any input is at most $\delta$) computed at the cost of an increase in the circuit depth. 
Von Neumann's results have been advanced over the years by a number of authors, notably Pippenger \cite{pippenger1988reliable}, Feder \cite{feder1989reliable}, Evans and Schulman \cite{evans1999signal}. 
The same question about noisy computation may be similarly posed for binary neural networks, which are known to be able to represent any Boolean functions \cite{anthony2003boolean}. In this paper, we leverage information theory to prove that, in analogy to Von Neumann's result, an increasing number of neurons is the cost to pay for using noisy neurons to perform computation. 

The analytic tool that is central to obtaining these results is the data processing inequality and its tighter variants. The data processing inequality (DPI) is a fundamental result of information theory. It formalizes the intuition that local data post-processing cannot increase mutual information. It is formulated as a contraction of mutual information for a Markov chain, which is  $I(X;Z)/I(X;Y)\le 1$ for $X \rightarrow Y \rightarrow Z$. When the channel $Y\rightarrow Z$ can be  characterized, channel dependent improvements of the DPI may be derived: $I(X;Z)/I(X;Y)\le \eta$ with $\eta<1$. These improved inequalities are called strong data processing inequalities (SDPIs) or sometimes quantified data processing inequalities. 

The DPI and SDPIs are particularly useful for proving impossibility results and find applications in a wide range of fields such as physics~\cite{dobrushin1970definition}, finance~\cite{erkip1998efficiency} and differential privacy~\cite{duchi2013local}. Many of the results are obtained for Binary Symmetric Channels (BSCs), which is the most common scenario to consider. Regarding noisy computation, Pippenger~\cite{pippenger1988reliable} first derived an SDPI for BSCs and applied it to show an upper bound on the maximum tolerable component noise in Boolean formulas. Evans and Schulman~\cite{evans1999signal} improved the SDPI bound of Pippenger and used it to obtain lower bounds on the complexity of reliable circuits with noisy components. 

In this paper, we derive an SDPI for $Y\rightarrow Z$ being a multiple-input and multiple-output noisy channel. Our SDPI strictly improves the tightness of the end-to-end information contraction bound from the Evans and Schulman method~\cite{evans1999signal}; more importantly, as our SDPI allows general channels other than BSCs being considered in Ref.~\cite{evans1999signal}, it can be applied to a wide range of different problems. 
As a first application, we apply this result to neural networks composed of noisy binary neurons. Our information contraction bound has a network width-depth trade-off relationship, which is then combined with network size complexity results to obtain lower bounds on the number of noisy neurons for reliable computation of Boolean functions such as the parity. 

As a second application, we further apply our SDPI derived results to fault-tolerant cellular automata~\cite{gacs2001reliable,gray2001reader,toom1974nonergodic}. We obtain bounds on the error correction overhead and the relaxation time. These findings can potentially shed new light on a wide range of topics, including phase transitions in statistical physics~\cite{liggett2012interacting,wolfram1983statistical,grinstein1985statistical} and stability of time-crystals~\cite{yao2020classical}.

The paper is organized as follows. In Section~\ref{Sec:def}, we introduce the basic definitions and notations. Section~\ref{Sec:SDPI} contains the major theorem of our SDPI and its variants applicable to different scenarios. The application of our SDPI to noisy binary neural networks is presented in section~\ref{Sec:NN}, while the case to fault-tolerant cellular automata in Section~\ref{Sec:CA}. Finally, we conclude in Section~\ref{Sec:conclusion} with a summary and some future directions.

\section{Definitions and notations}
\label{Sec:def}
Without loss of generality, a noisy binary neuron considered in this paper is defined as follows. For a noisy binary neuron that has $n$ synapses $(x_1,...,x_n)\in \left\{0,1\right\}^n$ with the corresponding synaptic weights $(w_1,...,w_n) \in {\rm I\!R}^n$ and a bias $\theta \in {\rm I\!R}$, the output of the neuron is a stochastic function
\begin{equation}
\hat{y}\,=\, \begin{cases}
  y \;\;\mathrm{with\;probability\;}1-p\\    
  (1-y) \;\;\mathrm{with\;probability\;}p
\end{cases},
\end{equation}
where
\begin{equation}
y\,=\,\mathrm{sgn}\left( \sum_{i=1}^n w_i x_i + \theta \right).
\end{equation}
And $\mathrm{sgn}$ is the sign function defined as
\begin{equation}
\mathrm{sgn}(\alpha)\,=\,\begin{cases}
  1 \;\;\mathrm{if\;}\alpha \ge0\\    
  0 \;\;\mathrm{Otherwise}
  \end{cases}.
\end{equation}
A schematic of this model is shown in Figure \ref{fig:Noisy_neuron}. A noisy binary neuron can be decomposed into a noiseless binary neuron followed by a noisy channel (denoted by NC in the schematic). With our definition above, the noisy channel is a BSC with bit flip probability $p$ (sometimes acronym-ed $\mathrm{BSC}(p)$). A noisy binary neuron is said to be $\xi$-noisy when the neuron misfiring probability $p$ is equal to $\xi$. 
A noiseless binary neuron described above is also called a threshold gate in some literature and is a basic unit of the class Threshold Circuits (\textbf{TC}) in circuit complexity theory \cite{impagliazzo1997size}.
\begin{figure}[tbp]
\begin{center}
\includegraphics[width=1.\linewidth]{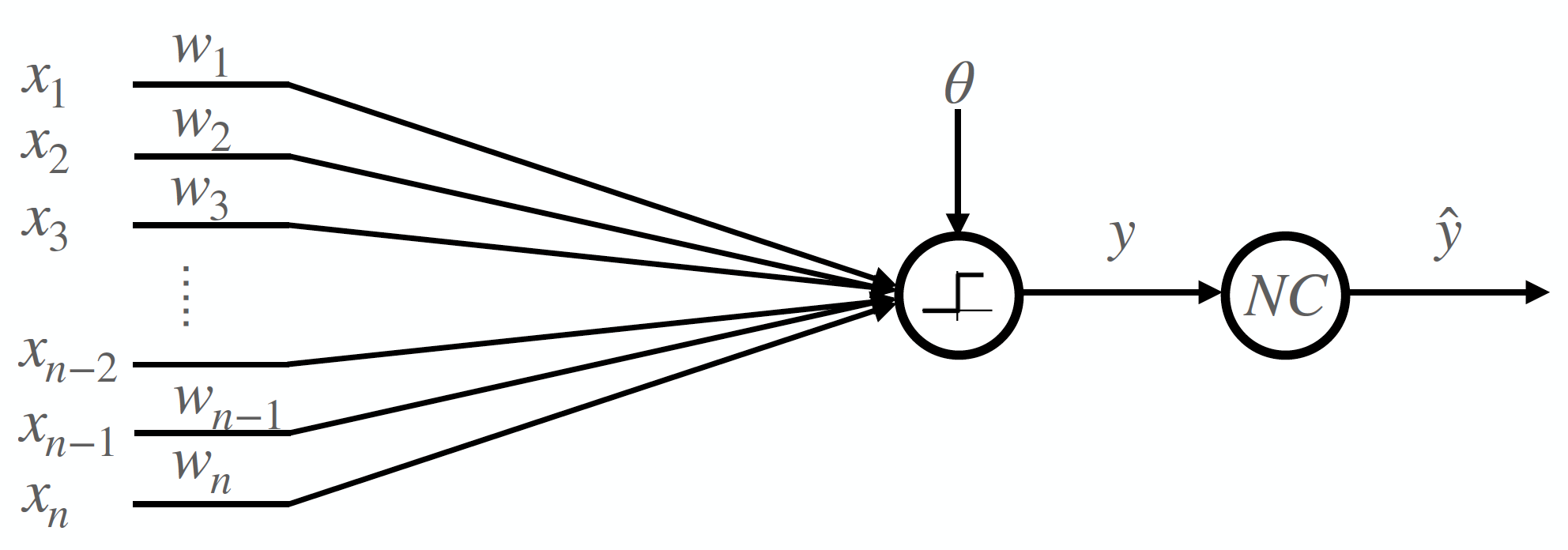}
\caption{Schematic of a noisy binary neuron with synapses $(x_1,...,x_n)$, synaptic weights $(w_1,...,w_n)$, bias $\theta$, binary step function activaition and noisy output $\hat{y}$.}\label{fig:Noisy_neuron}
\end{center}
\end{figure}

Even though so far we have only seen BSCs in our model, the key result of this paper is to prove an SDPI for a channel with an arbitrary number of input and output letter alphabets, which will then be used to improve bounds on estimating information contraction in a network composed of noisy binary neurons. 
Next, we give the definition of an \textit{n-input, m-output noisy channel} $X\rightarrow Y$. 
\begin{definition}\label{def:noisy_channel}
An n-input, m-output noisy channel is a mapping between a categorical random variable on $n$ elements (or an $n$-letter alphabet) $X$ and a categorical random variable on $m$ elements (or an $m$-letter alphabet) $Y$. A noisy channel is characterized by its transition matrix $\mathbf{A}$ such that $p_Y = p_X \mathbf{A}$, where $p_X$ and $p_Y$ are probability distributions over their letter alphabets represented as row vectors.
\end{definition}
Following our definition, $\mathbf{A}$ is an $n$ by $m$ matrix whose $(i,j)$-th entry $a_{i,j}$ is the probability that input $i$ is transformed into output $j$ by the channel. The entries of $A$ are non-negative and the rows sum to $1$: $a_{i,j}\ge 0$ and $\sum_{j} a_{i,j} = 1$.

We hereby give the definition of a simply layered network, which is a constraint on neuron connectivity that we use throughout this paper.
\begin{definition}
A neural network is said to be simply layered if its inputs and neurons can be partitioned into subsets called layers, such that every synaptic connection happens between adjacent layers and inputs only appear at layer $0$.
\end{definition}
The above definition describes the most commonly studied neural network architectures in a theoretical setting, namely the regular feed-forward networks with hidden layers but no skip connections allowed between non-adjacent layers.

Our main result will be concerning the mutual information
\begin{equation*}
I(X;Y)=H(X)+H(Y)-H(X,Y)=H(Y)-H(Y|X)
\end{equation*} 
between two random variables $X$ and $Y$, where the Shannon entropy $H(X)=-\sum_x p_X(x)\log\left(p_X(x)\right)$ can be calculated from the probability distribution $p_X$. We use the natural logarithm in this paper, however, our results will not depend on this choice. Here
\begin{equation*}
H(Y|X)=H(X,Y)-H(X)=\sum_{x}p_{X}(x)H(Y|X=x)
\end{equation*}
is the conditional entropy, which can be calculated from the conditional probability $p_{Y|X=x}(\cdot)$.

\section{A Strong Data Processing Inequality for general $n$-input $m$-output noisy channels}
\label{Sec:SDPI}
In this section, we prove our SDPI of Theorem~\ref{theo:1} and apply it to a network layer composed of noisy binary neurons, with Proposition~\ref{prop:1} for independent noise and Proposition~\ref{prop:2} for weakly-correlated noise.

Consider a Markov chain $X\rightarrow Y \rightarrow Z$, our objective is to bound the mutual information contraction ratio $I(X;Z)/I(X;Y)$ for $Y \rightarrow Z$ being an $n$-input, $m$-output noisy channel defined in Definition \ref{def:noisy_channel}, as a function of the characteristics of $Y \rightarrow Z$ alone. Previously, a bound has been obtained by Evans and Schulman \cite{evans1999signal} for $Y \rightarrow Z$ being a BSC with similar earlier results found in the more mathematically motivated work of Ahlswede and G\'{a}cs \cite{ahlswede1976spreading}. Theorem \ref{theo:1} in our paper generalizes these results to a general $n$-input, $m$-output noisy channel. 
\begin{theorem}
\label{theo:1}
\textbf{Strong Data Processing Inequality (SDPI)} X, Y and Z are categorical random variables following a Markov chain $X\rightarrow Y \rightarrow Z$, where $Y \rightarrow Z$ is an n-input, m-output noisy channel with a transition matrix $\mathbf{A}=(a_{i,j})$. Then
\begin{equation}
\label{eqn:SDPI}
\frac{I(X;Z)}{I(X;Y)}\le 1 - \min_{(k,\ell),k\ne \ell}\left(\sum_{j=1}^{m}\sqrt{a_{k,j} a_{\ell,j}}\right)^2.
\end{equation}
\end{theorem}
In order to prove the above theorem, we begin by introducing a lemma for generic concave functions admitting second derivatives. This lemma relates the maximum ratio of function evaluation to their second derivatives or Hessian matrices.
\begin{customlemma}{1}\label{lemma:1}
Functions $f$ and $g$ are concave functions that have second derivatives in the interior of $n$-dimensional probability simplex $\mathring{\mathcal{P}}_n=\left\{(x_0,...,x_n)\in {\rm I\!R}^{n+1}\;|\;\sum_{0\le i\le n} x_i=1, x_i>0\right\}$. $\forall \{ \boldsymbol{x}_i\}_{1\le i \le d}\in\mathring{\mathcal{P}}^d_n$ and $\sum_{i=1}^{d}p_i=1, p_i\ge0$ for any integer $d\ge 2$. Define functions
\begin{align}
f_2(\boldsymbol{x}_1,...,\boldsymbol{x}_d;p_1,...,p_d) = \sum_{i=1}^{d}p_i f(\boldsymbol{x}_i)-f(\sum_{i=1}^{d}p_i\boldsymbol{x}_i),
\\
g_2(\boldsymbol{x}_1,...,\boldsymbol{x}_d;p_1,...,p_d) = \sum_{i=1}^{d}p_i g(\boldsymbol{x}_i)-g(\sum_{i=1}^{d}p_i\boldsymbol{x}_i).
\end{align}
Then
\begin{equation}
\label{eqn:1}
\begin{aligned}
&\sup_{\substack{\forall (\boldsymbol{x}_1,...,\boldsymbol{x}_d)\in \mathring{\mathcal{P}}_{n}^d\\
\forall(p_1,...,p_d)\in\mathcal{P}_{d-1}}}
\frac{f_2(\boldsymbol{x}_1,...,\boldsymbol{x}_d;p_1,...,p_d)}{g_2(\boldsymbol{x}_1,...,\boldsymbol{x}_d;p_1,...,p_d)}
\\
&=\sup_{\substack{\forall \boldsymbol{c}\in {\rm I\!R}^{n}\\
\forall \boldsymbol{x}\in \mathring{\mathcal{P}}_n}}
\frac{\boldsymbol{c}^{\intercal}\mathbf{H}_{f}(\boldsymbol{x})\boldsymbol{c}}{\boldsymbol{c}^{\intercal}\mathbf{H}_{g}(\boldsymbol{x})\boldsymbol{c}},
\end{aligned}
\end{equation}
where $\mathbf{H}_{f}(\boldsymbol{x})$ and $\mathbf{H}_{g}(\boldsymbol{x})$ are Hessian matrices of functions $f$ and $g$ evaluated at  $\boldsymbol{x}$.
\end{customlemma}
We leave the proof of Lemma \ref{lemma:1} in Appendix~\ref{App:lemma1} so that we can focus on the main proof of Theorem \ref{theo:1}. The main idea of this proof is to turn the problem of upper bounding the ratio of mutual information into a problem of upper bounding the ratio of two quadratic forms obtained using Lemma \ref{lemma:1}. Then we finish the proof with a mathematical induction on the number of input letter alphabet $n$. The proof is rather technical and some readers may wish to skip to the end of it.
\begin{proof}
Recall that our goal is to find an upper bound strictly smaller than $1$ for the ratio $I(X;Z)/I(X;Y)$ with $Y\rightarrow Z$ being an $n$-input, $m$-output noisy channel.


As we will deal with $n$-element categorical random variables, it is convenient to introduce the discrete entropy function
\begin{equation}
g(p_1,...,p_n) = -\sum_{i=1}^{n} p_i\log(p_i),
\end{equation}
which is a function that maps a probability vector in the $n$-dimensional probability simplex $(p_1,...,p_n)\in\mathcal{P}_{n-1}$ to a positive real number.
$g$ is concave and admits well-defined second derivatives in the interior of $\mathcal{P}_{n-1}$: $\mathring{\mathcal{P}}_{n-1}$ where $p_i>0$, $\forall i$. By staying in the interior of the probability simplex, we only need to deal with non-degenerate cases (the true $n$-element cases). Degenerate cases where some probabilities are zero can be treated as cases with smaller $n$, which is handled in a mathematical induction over $n$.
For convenience, we also define the function
\begin{align}
f(p_1,...,p_n)&=g\left(\left(p_1,...,p_n\right)\mathbf{A}\right), 
\end{align}
where $\mathbf{A}$ is the $n$ by $m$ transition matrix of the noisy channel $Y\rightarrow Z$. As $g$ is strongly concave and concavity is preserved by a linear transformation of its argument (the negativeness of its Hessian is preserved), hence $f$ is also a concave function and admits second derivatives in $\mathring{\mathcal{P}}_{n-1}$.

It can be shown that $I(X;Z)=-f_2(p_{Y|X};p_{X})$ and $I(X;Y)=-g_2(p_{Y|X};p_{X})$, with $f_2$ and $g_2$ defined in Lemma \ref{lemma:1} and the above definitions of $f$ and $g$. We give a proof for the first identity and the other one can be proven similarly.
\begin{align}
&I(X;Z)= H(Z)-H(Z|X)\nonumber\\
&=g(p_Y\mathbf{A})-\sum_{x}p_{X}(x)H(Z|X=x) \nonumber\\
&=g\left(\sum_{X=x}p_{X}\left(x\right)\left[p_{Y|X=x}\left(y\right)\mathbf{A}\right]\right)
\nonumber
\\
&
\qquad \qquad \qquad
-\sum_{X=x}p_{X}\left(x\right)g\left(p_{Y|X=x}\left(y\right)\mathbf{A}\right)\nonumber \\
&=-f_2(p_{Y|X};p_{X}).
\end{align}
Functions $f$ and $g$ are concave and admit second derivatives in $\mathring{\mathcal{P}}_{n-1}$. We can therefore apply Lemma \ref{lemma:1} with $f,g$ and $f_2,g_2$ defined above:
\begin{equation}\label{eqn:main}
\frac{I(X;Z)}{I(X;Y)}=\frac{f_2(p_{Y|X};p_{X})}{g_2(p_{Y|X};p_{X})}\le
\sup_{\substack{\forall \boldsymbol{c}\in {\rm I\!R}^{n-1}\\
\forall \boldsymbol{p}\in \mathring{\mathcal{P}}_{n-1}}}
\frac{\boldsymbol{c}^{\intercal}\mathbf{H}_{f}(\boldsymbol{p})\boldsymbol{c}}{\boldsymbol{c}^{\intercal}\mathbf{H}_{g}(\boldsymbol{p})\boldsymbol{c}}.
\end{equation}
Here, we have successfully transformed the original problem into upper bounding the ratio of two quadratic forms defined by Hessian matrices of $f$ and $g$. 
Observe that the right hand side term does not depend on the exact dimension nor the probability distribution of $X$, removing lots of complexity and implying that $X$ can be arbitrarily sized. 
By extending Lemma \ref{lemma:1} to the measure-theoretic form, it can be shown that the result holds even for $X$ being a continuous random variable, only the discrete case is needed for the scope of this paper. 
We proceed to calculate the $(n-1)\times(n-1)$ Hessian matrices $\mathbf{H}_{f}(\boldsymbol{p})$ and $\mathbf{H}_{g}(\boldsymbol{p})$ evaluated for a probability vector $\boldsymbol{p}\in \mathring{\mathcal{P}}_{n-1}$.
We first evaluate $\mathbf{H}_{g}(\boldsymbol{p})$,
with $p_n=1-\sum_{i=1}^{n-1}p_i$ and we treat $(p_1,...,p_{n-1})$ as independent variables.
Taking second derivatives, we have $\partial^2 g/\partial p_i \partial p_j = -1/p_n$ for $\forall (i,j),i\ne j$ and $\partial^2 g/\partial p_i^2 = -1/p_i-1/p_n$. Hence
\begin{equation}
\mathbf{H}_{g}(\boldsymbol{p})=-
\begin{pmatrix}
\displaystyle\frac{p_1+p_n}{p_1 p_n}  & \displaystyle\frac{1}{p_n}      &  \dots   & \displaystyle\frac{1}{p_n}  \\
\displaystyle\frac{1}{p_n}   & \displaystyle\frac{p_2+p_n}{p_1 p_n} & \ddots   & \vdots \\
\vdots  &  \ddots       & \ddots  & \displaystyle\frac{1}{p_n} \\
\displaystyle\frac{1}{p_n} & \dots & \displaystyle\frac{1}{p_n}    & \displaystyle\frac{p_{n-1}+p_n}{p_{n-1}p_n}
\end{pmatrix}.
\end{equation}
The calculation of $\mathbf{H}_{f}(\boldsymbol{p})$ is done in a similar way by remarking that
\begin{equation}
\begin{aligned}
&f(\boldsymbol{p})=g\left(\left(p_1,p_2,...,p_n\right)\mathbf{A}\right) \\
&=g\left(\sum_{i=1}^{n}p_i a_{i,1}\;,\sum_{i=1}^{n}p_i a_{i,2}\;,...\;,\sum_{i=1}^{n}p_i a_{i,m}\right) \\
&=-\sum_{j=1}^m\left(\sum_{i=1}^n p_i a_{i,j}\right)
\log\left(\sum_{i=1}^n p_i a_{i,j}\right).
\end{aligned}
\end{equation}
Additionally, with $p_n=1-\sum_{i=1}^{n-1} p_i$, we have $\sum_{i=1}^n p_i a_{i,j}=a_{n,j}+\sum_{i=1}^{n-1}p_i(a_{i,j}-a_{n,j})$.
The second derivatives of $f$ are evaluated as 
\begin{equation} \label{eqn:partial_f}
\frac{\partial^2 f}{\partial p_\ell \partial p_k}=
-\sum_{j=1}^m\frac{(a_{k,j}-a_{n,j})(a_{\ell,j}-a_{n,j})}{\sum_{i=1}^n p_i a_{i,j}},
\end{equation}
for $1\le k,\ell \le n-1$. The matrix form of $\mathbf{H}_{f}(\boldsymbol{p})$ can be found in Appendix~\ref{App:theorem1}. 

We prove the final result by an induction on $n$. The smallest $n$ we consider is $n=2$ ($n=1$ is an uninteresting degenerate case), which corresponds to a binary variable $Y$ and a $2$-input, $m$-output noisy channel. Evans has previously given a proof \cite{evans1994information} for this special case in his doctoral thesis, his proof is via a geometric argument by computing the divergence of two infinitesimally separated distributions as an Euclidean distance. Our proof is as follows.

When $n=2$, the Hessian matrices $\mathbf{H}_{f}(\boldsymbol{p})$ and $\mathbf{H}_{g}(\boldsymbol{p})$ are simple scalars and 
\begin{equation}
\begin{aligned}
\frac{I(X;Z)}{I(X;Y)}\le 
\sup_{\substack{\forall c\in {\rm I\!R}\\
\forall \boldsymbol{p}\in \mathring{\mathcal{P}}_{1}}}
\frac{c^2\mathbf{H}_{f}(\boldsymbol{p})}{c^2\mathbf{H}_{g}(\boldsymbol{p})}
=\sup_{\substack{
\forall \boldsymbol{p}\in \mathring{\mathcal{P}}_{1}}}
\sum_{j=1}^{m}\frac{p_1 p_2(a_{1,j}-a_{2,j})^2}{p_1 a_{1,j} + p_2 a_{2,j}}.
\end{aligned}
\end{equation}
It suffices to find an upper bound for $\sum_{j=1}^{m}{p_1 p_2(a_{1,j}-a_{2,j})^2}/{\left(p_1 a_{1,j} + p_2 a_{2,j}\right)}$. $\mathbf{A}$ is a transition matrix so its rows sum to $1$: $\sum_{j=1}^m a_{i,j}=1,\forall i$. Thus we have the following identities: $1=\sum_{j=1}^m (p_1 a_{2,j}+p_2 a_{1,j})=\sum_{j=1}^m (p_1 a_{1,j}+p_2 a_{2,j})$, therefore
\begin{align}
&1-\sum_{j=1}^{m}\frac{p_1 p_2(a_{1,j}-a_{2,j})^2}{p_1 a_{1,j} + p_2 a_{2,j}}
\nonumber\\
&=\sum_{j=1}^m (p_1 a_{2,j}+p_2 a_{1,j})-\sum_{j=1}^{m}\frac{p_1 p_2(a_{1,j}-a_{2,j})^2}{p_1 a_{1,j} + p_2 a_{2,j}} \nonumber\\
&=\sum_{j=1}^m \frac{p_1^2 a_{1,j}a_{2,j}+p_2^2 a_{1,j}a_{2,j}+2p_1 p_2 a_{1,j}a_{2,j}}{p_1 a_{1,j} + p_2 a_{2,j}} \nonumber \\
&=\sum_{j=1}^m \frac{a_{1,j}a_{2,j}}{p_1 a_{1,j} + p_2 a_{2,j}} \nonumber \\
&=\left[\sum_{j=1}^m \left(\sqrt{\frac{a_{1,j}a_{2,j}}{p_1 a_{1,j} + p_2 a_{2,j}}}\right)^2\right]
\left[\sum_{j=1}^m \left(\sqrt{p_1 a_{1,j}+p_2 a_{2,j}}\right)^2\right] \nonumber \\
&\ge \left(\sum_{j=1}^m \sqrt{{a_{1,j}a_{2,j}}}\right)^2.
\end{align}
We have used the Cauchy-Schwarz inequality at the last step. Rearranging terms gives 
\begin{equation}
\sum_{j=1}^{m}\frac{p_1 p_2(a_{1,j}-a_{2,j})^2}{p_1 a_{1,j} + p_2 a_{2,j}}\le 1-\left(\sum_{j=1}^m \sqrt{{a_{1,j}a_{2,j}}}\right)^2. 
\end{equation}
The theorem is proven for $n=2$.

Now suppose that Theorem \ref{theo:1} holds true up to a number of inputs $n$ to the noisy channel with $n\ge 2$. We prove our claim for $n+1$ inputs and conclude by an induction argument. First, notice that if one or more entries to the $n+1$ dimension probability $\boldsymbol{p}$ vector are zeros and we have a degenerate case, then it corresponds to a situation where the input number is strictly smaller than $n+1$, which is already handled by our inductive hypothesis. We only need to deal with non-degenerate cases where $\mathbf{H}_{f}(\boldsymbol{p})$ and $\mathbf{H}_{g}(\boldsymbol{p})$ are $n\times n$ matrices with well-defined entries.

Inequality \eqref{eqn:main} indicates that we need to find an upper bound for the ratio $\boldsymbol{c}^{\intercal}\mathbf{H}_{f}(\boldsymbol{p})\boldsymbol{c}/\boldsymbol{c}^{\intercal}\mathbf{H}_{g}(\boldsymbol{p})\boldsymbol{c}$, for notational simplicity, let $Q_f(\boldsymbol{c})$ denote the quadratic form $-\boldsymbol{c}^{\intercal}\mathbf{H}_{f}(\boldsymbol{p})\boldsymbol{c}$. Similarly, we define $Q_g(\boldsymbol{c})=-\boldsymbol{c}^{\intercal}\mathbf{H}_{g}(\boldsymbol{p})\boldsymbol{c}$. $Q_g(\boldsymbol{c})$ can be straightforwardly evaluated as 
\begin{align}\label{eqn:Q_g}
&Q_g(\boldsymbol{c}) =\sum_{s=1}^n c_s^2\frac{p_s+p_{n+1}}{p_s p_{n+1}}+\sum_{1\le s < t \le n} c_s c_t \frac{2}{p_{n+1}}  \nonumber\\
&=\sum_{s=1}^n \left(c_s^2+\sum_{t \ne s}c_s c_t\right)\frac{p_s+p_{n+1}}{p_s p_{n+1}}-\sum_{1\le s <t \le n} c_s c_t \frac{p_s + p_t}{p_s p_t}.
\end{align}
Notice that we have a series of identities obtained from the normalization properties of $\mathbf{A}$: 
\begin{align}
1&=\sum_{j=1}^m \frac{p_s a_{t,j}+p_{t} a_{s,j}}{p_s+p_{t}},\;1\le s<t\le n+1.
\end{align}
Multiply the corresponding terms in the right hand side of Equation \eqref{eqn:Q_g} by the identities (for the same $s$ and $t$ respectively) above, we can prove that
\begin{equation} \label{eqn:Q_g_int}
Q_g(\boldsymbol{c})=Q_f(\boldsymbol{c})+\sum_{1\le s<t \le n+1}Q_{s,t}(\boldsymbol{c})\sum_{j=1}^m
\frac{a_{s,j}a_{t,j}}{\sum_{i=1}^{n+1}p_i a_{i,j}},
\end{equation}
where $Q_{s,t}(\boldsymbol{c})$ is a quadratic form of $\boldsymbol{c}$. 
The details to obtain Equation \eqref{eqn:Q_g_int} can be found in Appendix~\ref{App:theorem1}.
We are going to prove that $Q_{s,t}(\boldsymbol{c})$ is a square term and thus positive.

We first investigate the terms $Q_{s,t}(\boldsymbol{c})$ for $1\le s<t \le n$, a careful bookkeeping of all the relevant terms gives
\begin{align}
Q_{s,t}(\boldsymbol{c})
=\left(\sqrt{\frac{p_t}{p_s}}c_s - \sqrt{\frac{p_s}{p_t}}c_t\right)^2 \ge 0.
\end{align}
Then we tackle the terms $Q_{s,n+1}(\boldsymbol{c})$ for $1\le s \le n$:
\begin{align}
&Q_{s,n+1}(\boldsymbol{c}) \nonumber \\
&=\left[c_s\left(\sqrt{\frac{p_s}{p_{n+1}}}+\sqrt{\frac{p_{n+1}}{p_{s}}}\right)+\sum_{t \ne s} \sqrt{\frac{p_s}{p_{n+1}}}c_t \right]^2 \ge 0.
\end{align}
And the Cauchy-Schwarz inequality gives for any $s$ and $t$
\begin{align}
&\sum_{j=1}^m
\frac{a_{s,j}a_{t,j}}{\sum_{i=1}^{n+1}p_i a_{i,j}} \nonumber \\
&=\sum_{j=1}^m
\left(\sqrt{\frac{a_{s,j}a_{t,j}}{\sum_{i=1}^{n+1}p_i a_{i,j}}}\right)^2
\sum_{j=1}^m \left(\sqrt{\sum_{i=1}^{n+1}p_i a_{i,j}}\right)^2 \nonumber \\
&\ge \left(\sum_{j=1}^m \sqrt{a_{s,j} a_{t,j}}\right)^2.
\end{align}
Since $Q_{s,t}$ and $Q_{s,n+1}$ are all positive, we apply this inequality in Equation \eqref{eqn:Q_g_int}, and 
\begin{align} \label{eqn:before_last}
&Q_g(\boldsymbol{c}) \nonumber \\
&\ge Q_f(\boldsymbol{c})+\sum_{1\le s<t \le n+1}Q_{s,t}(\boldsymbol{c})
\left(\sum_{j=1}^m \sqrt{a_{s,j} a_{t,j}}\right)^2.
\nonumber \\
&\ge Q_f(\boldsymbol{c})+
\sum_{1\le s<t \le n+1}Q_{s,t}(\boldsymbol{c}) \min_{(k,\ell),k\ne \ell} \left(\sum_{j=1}^m \sqrt{a_{k,j} a_{\ell,j}}\right)^2.
\end{align}
Next we prove an identity that $Q_g(\boldsymbol{c})=\sum_{1\le s<t \le n+1}Q_{s,t}(\boldsymbol{c})$.
Equation \eqref{eqn:Q_g_int} is an identity that holds for any $\mathbf{A}$. In particular, let us consider the case where $a_{s,j}=a_{t,j}=a_j\;\forall(s,t)$, meaning all rows of $\mathbf{A}$ are equal. It is easy to see that $Q_f(\boldsymbol{c})$ is zero in this case from Equation \eqref{eqn:partial_f}. Therefore
\begin{align}
Q_{g}(\boldsymbol{c})&=\sum_{1\le s<t \le n+1}Q_{s,t}(\boldsymbol{c})\sum_{j=1}^m
a_j
=
\sum_{1\le s<t \le n+1}Q_{s,t}(\boldsymbol{c})
\end{align}
Using this identity in Inequality \eqref{eqn:before_last} gives 
\begin{align}
Q_g(\boldsymbol{c})\ge Q_f(\boldsymbol{c}) + Q_g(\boldsymbol{c})
\min_{(k,\ell),k\ne \ell} \left(\sum_{j=1}^m \sqrt{a_{k,j} a_{\ell,j}}\right)^2.
\end{align}
$Q_g(\boldsymbol{c})=-\boldsymbol{c}^{\intercal}\mathbf{H}_{g}(\boldsymbol{p})\boldsymbol{c}$ and $\mathbf{H}_{g}(\boldsymbol{p})$ is negative definite by the strong concavity of $g$. $Q_g(\boldsymbol{c})$ is therefore positive and
\begin{equation}
\frac{\boldsymbol{c}^{\intercal}\mathbf{H}_{f}(\boldsymbol{p})\boldsymbol{c}}{\boldsymbol{c}^{\intercal}\mathbf{H}_{g}(\boldsymbol{p})\boldsymbol{c}}
=\frac{Q_f(\boldsymbol{c})}{Q_g(\boldsymbol{c})}\le 1-\min_{(k,\ell),k\ne \ell} \left(\sum_{j=1}^m \sqrt{a_{k,j} a_{\ell,j}}\right)^2.
\end{equation}
Hence, Inequality \eqref{eqn:main} gives
\begin{equation}
\frac{I(X;Z)}{I(X;Y)} \le 1-\min_{(k,l),k\ne \ell} \left(\sum_{j=1}^m \sqrt{a_{k,j} a_{\ell,j}}\right)^2.
\end{equation}
We therefore conclude the proof by induction.
\end{proof}
The results we have obtained so far are generic and next we apply it to estimate information contraction in a network composed of noisy binary neurons.
\begin{prop}
\label{prop:1}
\textbf{Independent neuron noise.}
Consider a layer composed of $n$ $\xi$-noisy binary neurons that independently fail with probability $\xi$, let $X$ be the input to the binary neural network ($X$ is random bit string of an arbitrary number of bits), $Y$ be the pre-noise output of the neurons and $Z$ be the output after noise. Then
\begin{equation}
\frac{I(X;Z)}{I(X;Y)}\le 1 - (4\xi-4\xi^2)^n.
\label{contraction_single_layer}
\end{equation}
\end{prop}
\begin{proof}
First, as the output $Z$ is independent of $X$ given $Y$, so we have $X\rightarrow Y \rightarrow Z$ as a Markov chain, where the mapping $Y \rightarrow Z$ is a noisy channel of $2^n$ inputs and $2^n$ outputs. The entries of transition matrix $\mathbf{A}$ are binomial terms $\xi^k(1-\xi)^{n-k}$ with $0\le k\le n$ (there are $2^n$ terms with repetitions for all possible combinations, each row of $\mathbf{A}$ is a permutation of these terms). Notice that $0\le  \xi<0.5$, thus $1-\xi>\xi$. We can rank all the binomial terms in descending order: $(1-\xi)^n>(1-\xi)^{n-1}\xi>...>\xi^{n-1}(1-\xi)>\xi^n$. Apply rearrangement inequality and monotonicity of the square root function, for $\forall(k,\ell)$, the minimum of $\sum_{j=1}^{2^n}\sqrt{a_{k,j} a_{\ell,j}}$ is achieved by arranging binomial terms of $\{a_{k,j}\}_j$ and $\{a_{\ell,j}\}_j$ in a reversed order, leading to each product $a_{k,j} a_{\ell,j}$ equaling $\xi^n(1-\xi)^n$, i.e., $\sum_{j=1}^{2^n}\sqrt{a_{k,j} a_{\ell,j}}\ge \sum_{j=1}^{2^n} \sqrt{\xi^n(1-\xi)^n}=(4\xi-4\xi^2)^{n/2}$. Last, applying Theorem \ref{theo:1} gives $I(X;Z)\le I(X;Y)[1-(4\xi-4\xi^2)^n]$. 
\end{proof}
As a sanity check, we have $1-(4\xi-4\xi^2)^n=1$ when $\xi=0$ and $1-(4\xi-4\xi^2)^n=0$ when $\xi=0.5$. These values correspond to the perfectly noiseless and the perfectly noisy cases and as one would expect, the upper bound is one in the perfectly noiseless scenario and zero in the opposite. Proposition \ref{prop:1} is generic enough such that by swapping noisy binary neurons with noisy binary gates or Bayesian dependencies, it can also be applied to noisy circuits and Bayesian networks.  Additionally, $X$ can be restricted to a subset of the network input, the proof still holds.

When $n=1$ and $X$ is a random bit, Proposition \ref{prop:1} gives the result of Theorem 1 in Evans and Schulman \cite{evans1999signal}. When $n>1$, Evans and Schulman proved a result in their paper (Lemma 2) which can be used to upper bound end-to-end information contraction. However, in the following we show that our contraction bound in Inequality~\eqref{contraction_single_layer} is strictly tighter than the one obtainable from the resutls of Evans and Schulman for $n>1$. 
In Figure \ref{fig:ES99_comparison}, we give a simple example of a network which has a binary random variable $X$ as input, followed by three $\xi$-noisy binary neurons. Let $Y$ be the pre-noise output of these neurons, $Z$ be the output after noise. Equivalently, we can consider that $Z$ is the output of $Y$ going through three independent BSCs of probability $\xi$, as illustrated in Figure \ref{fig:ES99_comparison}. Evans-Schulman method estimates the end-to-end information contraction upper bound to be $3[1-(4\xi-4\xi^2)]$ while Proposition \ref{prop:1} gives $1-(4\xi-4\xi^2)^3$. More generally speaking, Evans-Schulman method gives an upper bound of $n\eta$ and our estimate gives $[1-(1-\eta)^n]$ with $\eta = 1-(4\xi-4\xi^2)$. We have $n\eta-[1-(1-\eta)^n]=0$ while $\eta=0$, taking its derivative with respect to $\eta$ gives $n-n(1-\eta)^{n-1}$, which is positive for any $0<\eta<1$. Hence $n\eta > 1-(1-\eta)^n$ for any $0<\eta<1$, the upper bound given by our Proposition \ref{prop:1} is always tighter than the one obtained from the Evans-Schulman method. In the limit where $\eta \rightarrow 0$, the Evans-Schulman upper bound captures first order behavior of $[1-(1-\eta)^n]$.

\begin{figure}[th]
\begin{center}
\includegraphics[width=.88\linewidth]{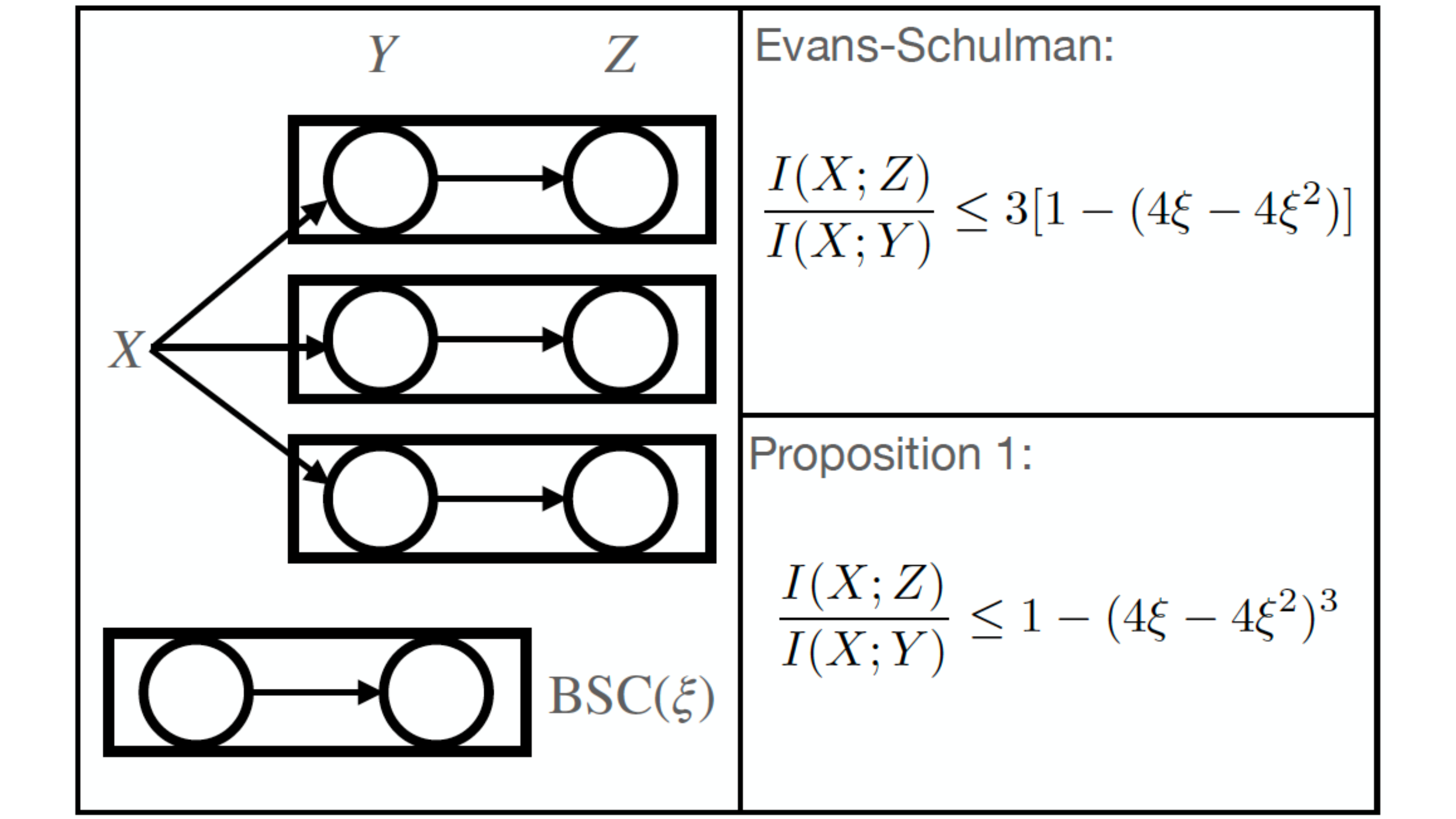}
\includegraphics[width=1.\linewidth]{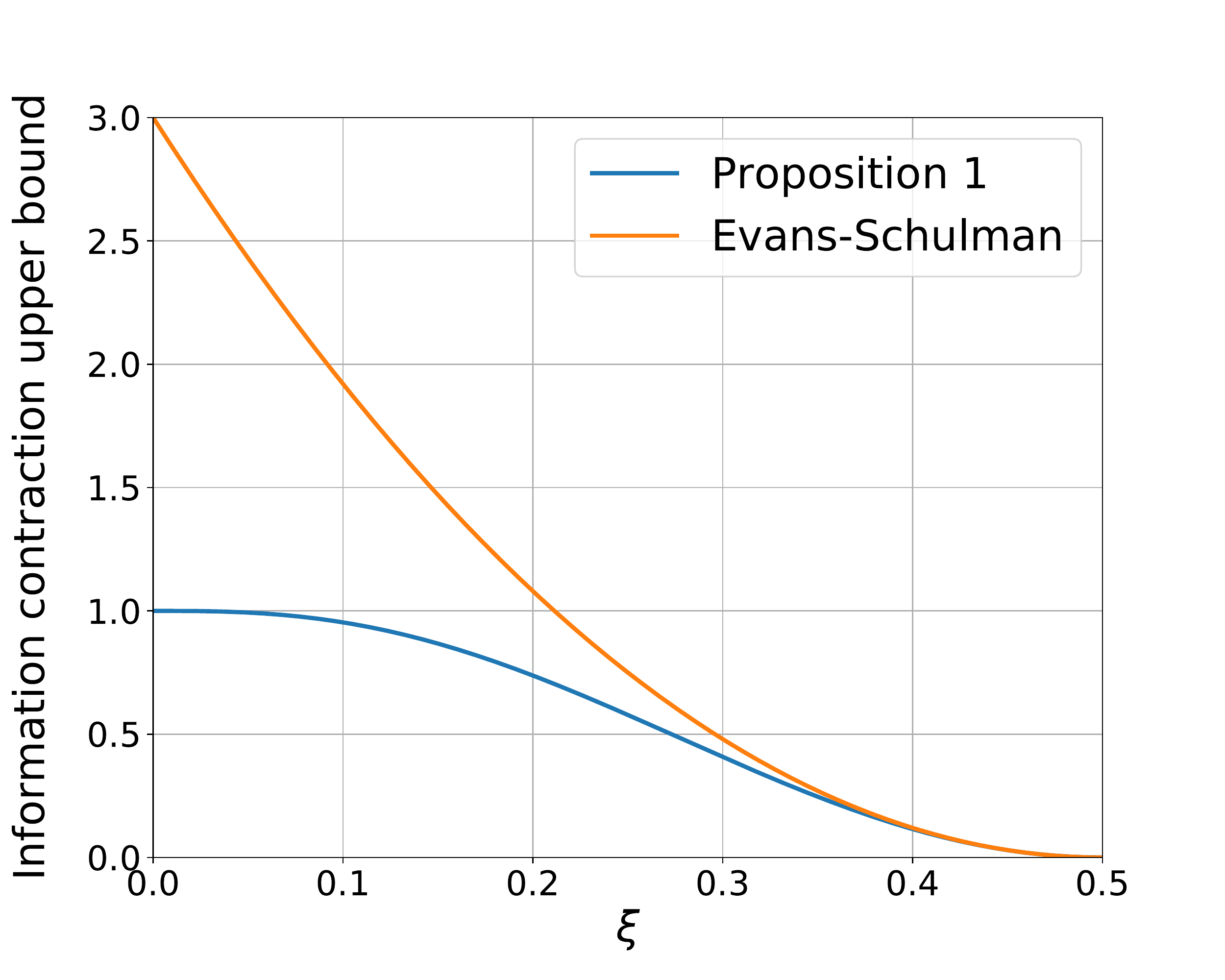}
\caption{Comparison of Evans-Schulman upper bound \cite{evans1999signal} (Lemma 2) with ours for a simple example network shown by the upper plot, the lower plot shows that our upper bound given by Proposition \ref{prop:1} is strictly tighter.}\label{fig:ES99_comparison}
\end{center}
\end{figure}

The authors would like to point out that similar results to Proposition \ref{prop:1} can be obtained from Theorem 5 of Polyanskiy and Wu \cite{polyanskiy2017strong}, which has a very different proof from ours and uses percolation theory. 
Our proof using Theorem \ref{theo:1} also has the advantage of being applicable to more general situations where noise or failure is correlated between different neurons and not independent. 
Correlated neuron noise has been observed in biological \cite{rumyantsev2020fundamental} and physical science \cite{semenova2019fundamental} contexts. Next we give some useful results derived from Theorem \ref{theo:1} taking into account the effect of weakly-correlated layer-wise neuron noise. 

The most general form of Theorem \ref{theo:1} is not analytically tractable and requires the calculation of each pair of row-wise inner products of a $2^n$ by $2^n$ matrix $\mathbf{A}$.
However, we can consider a simplified model of layer-wise noise correlation by decomposing the noise to a shared noise source and an independent noise source. 
The shared noise is characterized by the probability $\xi_1$ that all neurons in the layer fail together. 
Each neuron also has an independent noise source with an error probability of $\xi_2$. 
We consider that the two noise sources are combined in a successive way independent from one another. 
Intuitively speaking, first a biased coin is tossed to decide if all neurons of a layer will flip their output with probability $\xi_1$ and then each neuron in that layer tosses its own coin to decide if it will flip its own output again with probability $\xi_2$. 
We have the following result under these assumptions. 
\begin{prop}
\label{prop:2}
\textbf{Weakly-correlated layer-wise neuron noise.}
\newline
Consider a layer composed of $n$ noisy binary neurons that have a joint failure probability of $\xi_1$ and an independent individual failure probability of $\xi_2$ under our assumptions, the noise is weakly-correlated such that $\xi_1 \ll 1$ and $\xi_1 \ll \xi_2$. 
Let $X$ be the input to the binary neural network, $Y$ be the pre-noise output of the neurons and $Z$ be the output after noise. Then
\begin{equation} \label{eqn:eta_wc}
\frac{I(X;Z)}{I(X;Y)}\le 1 - [(4\xi_2-4\xi_2^2)^n+g(\xi_2,n)\xi_1]+\smallO(\xi_1),
\end{equation}
where 
\begin{equation} \label{eqn:g}
g(\xi_2,n)=2[\left(4\xi_2^2-4\xi_2+2\right)^n-\left(4\xi_2-4\xi_2^2\right)^n].
\end{equation}
\end{prop}
\begin{proof}
The overall idea of this proof is similar to that of Proposition \ref{prop:1}. 
The entries of matrix $\textbf{A}$ are $\{ (1-\xi_1)(1-\xi_2)^{n-k}\xi_2^k+\xi_1 \xi_2^{n-k}(1-\xi_2)^k\}_{0\le k \le n}$ with repetitions, taking into account the extra paths of transition introduced by the shared noise source. 
When $\xi_1 \rightarrow 0$, they converge to simple binomial terms $\{\xi_2^k(1-\xi_2)^{n-k} \}_{0\le k \le n}$ which are strongly ordered as shown previously in the proof of Proposition \ref{prop:1}. 
A continuity argument with respect to $\xi_1$ gives that the ordering is preserved for small enough $\xi_1$. 
Thus rearrangement inequality still applies and each term in the reversed order product is (to the leading order in $\xi_1$)
\begin{equation}
(1-\xi_2)^n\xi_2^n - 2\xi_1 (1-\xi_2)^n\xi_2^n
+ 2 \xi_1 (1-\xi_2)^{2(n-k)}\xi_2^{2k}.
\end{equation}
And the sum is lower bounded as (to the leading order in $\xi_1$)
\begin{align}
&\sum_{j=1}^{2^n}\sqrt{a_{k,j} a_{\ell,j}}
\ge
(1-\xi_1)(4\xi_2-4\xi_2^2)^{n/2}
+ \nonumber\\
&\frac{\xi_1}{(\xi_2-\xi_2^2)^{n/2}} \sum_{k=0}^n {n \choose k}[(1-\xi_2)^2]^{n-k}(\xi_2^2)^k \nonumber \\
&=(1-\xi_1)(4\xi_2-4\xi_2^2)^{n/2}
+ \frac{\xi_1}{(\xi_2-\xi_2^2)^{n/2}}
[2\xi_2^2 - 2\xi_2 + 1]^n.
\end{align}
Applying Theorem \ref{theo:1} with the above inequality proves the proposition.
\end{proof}
The information contraction upper bound in Proposition \ref{prop:2} can be compared with the bound in Proposition \ref{prop:1}. We have the following corollary comparing their magnitudes. 
\begin{corollary} \label{coro:1bis}
At the same per neuron noise level between the case of independent noise and the case of weakly-correlated noise, meaning $\xi_1 (1-\xi_2)+(1-\xi_1)\xi_2 = \xi$, 
let $\eta_\mathrm{ind}(\xi)$ denote the information contraction upper bound for independent neuron noise in Proposition \ref{prop:1} and let $\eta_\mathrm{wc}(\xi_1, \xi_2)$ denote the information contraction upper bound for weakly-correlated neuron noise in Proposition \ref{prop:2}. In the weakly-correlated noise limit where $\xi_1 \ll 1$ and $\xi_1 \ll \xi_2$, we have
\begin{equation}
    \eta_\mathrm{wc}(\xi_1, \xi_2)
\le
\eta_\mathrm{ind}(\xi).
\end{equation}
\end{corollary}
\begin{proof}
Plug the expression $\xi_1 (1-\xi_2)+(1-\xi_1)\xi_2 = \xi$ into Proposition \ref{prop:1} and expand $(4\xi-4\xi^2)^n$ to leading order in $\xi_1$. We have 
\begin{equation} \label{eqn:eta_ind}
    \eta_\mathrm{ind}(\xi) = 1 - [(4\xi_2-4\xi_2^2)^n+\tilde{g}(\xi_2,n)\xi_1]+\smallO(\xi_1),
\end{equation}
with $\tilde{g}(\xi_2,n)=4n(2\xi_2-1)^2(4\xi_2 - 4\xi_2^2)^{n-1}$. 
Function $g(\xi_2,n)$ can be factorized as
$g(\xi_2,n) = 4(4\xi_2^2-4\xi_2+1)[\sum_{i=1}^{n} (4\xi_2^2-4\xi_2+2)^{n-i}(4\xi_2 - 4\xi_2^2)^{i-1}]$ for $n\ge 1$.
Since $4\xi_2^2-4\xi_2+2 \ge 4\xi_2 - 4\xi_2^2$ for $0\le \xi_2 \le 0.5$, thus we have
$g(\xi_2,n)\ge4(2\xi_2-1)^2 n (4\xi_2 - 4\xi_2^2)^{n-1}=\tilde{g}(\xi_2,n)$ and the corollary is proven by comparing Equation \eqref{eqn:eta_ind} with Equation \eqref{eqn:eta_wc}.
\end{proof}

\begin{figure}[t]
\begin{center}
\includegraphics[width=1.\linewidth]{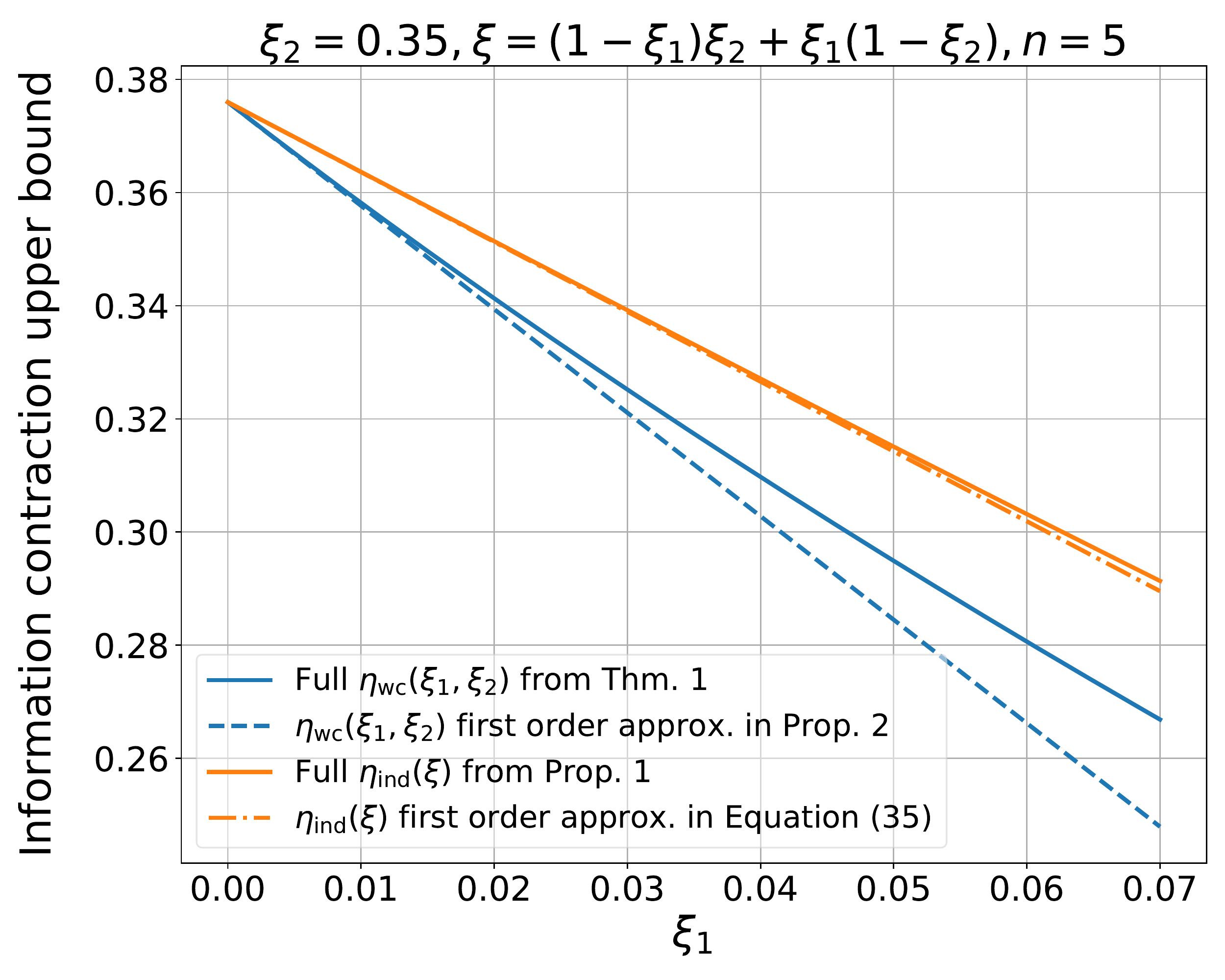}
\caption{A numerical comparison of the information contraction upper bound for independent noise and weakly-correlated layer-wise neuron noise with $\xi_2=0.35$ and $n=5$. Weakly-correlated noise results in a lower contraction ratio upper bound than independent noise as stated in Corollary \ref{coro:1bis}.}\label{fig:correlated}
\end{center}
\end{figure}
Here we give a numerical example to illustrate these results. 
Figure \ref{fig:correlated} shows a comparison between the case of independent noise and the case of weakly-correlated noise for a layer of $n=5$ neurons and $\xi_2=0.35$. 
It is numerically verifiable that the ordering of terms is preserved for $\xi_1$ up to $0.07$ and $\eta_\mathrm{wc}$ can be evaluated directly from Theorem \ref{theo:1} by calculating the reversed-order sum of products. 
Our leading order result in Proposition \ref{prop:2} gives a good approximation in the small $\xi_1$ limit. 
Weakly-correlated neuron noise leads to a lower end-to-end information contraction upper bound as stipulated by Corollary \ref{coro:1bis}.

\section{Application to noisy neural networks}
\label{Sec:NN}

In this section, we apply the SDPI results obtained in the previous section to study the end-to-end information contraction in a neural network. 
We will only have formal results for the case of independent noise in the rest of this paper, whenever we say that a component is $\xi$-noisy, it implies that its failure is independent.
Theorem \ref{thm:2} gives the main result by successively applying Proposition \ref{prop:1} to each layer.

\begin{figure}[t]
\begin{center}
\includegraphics[width=1.\linewidth]{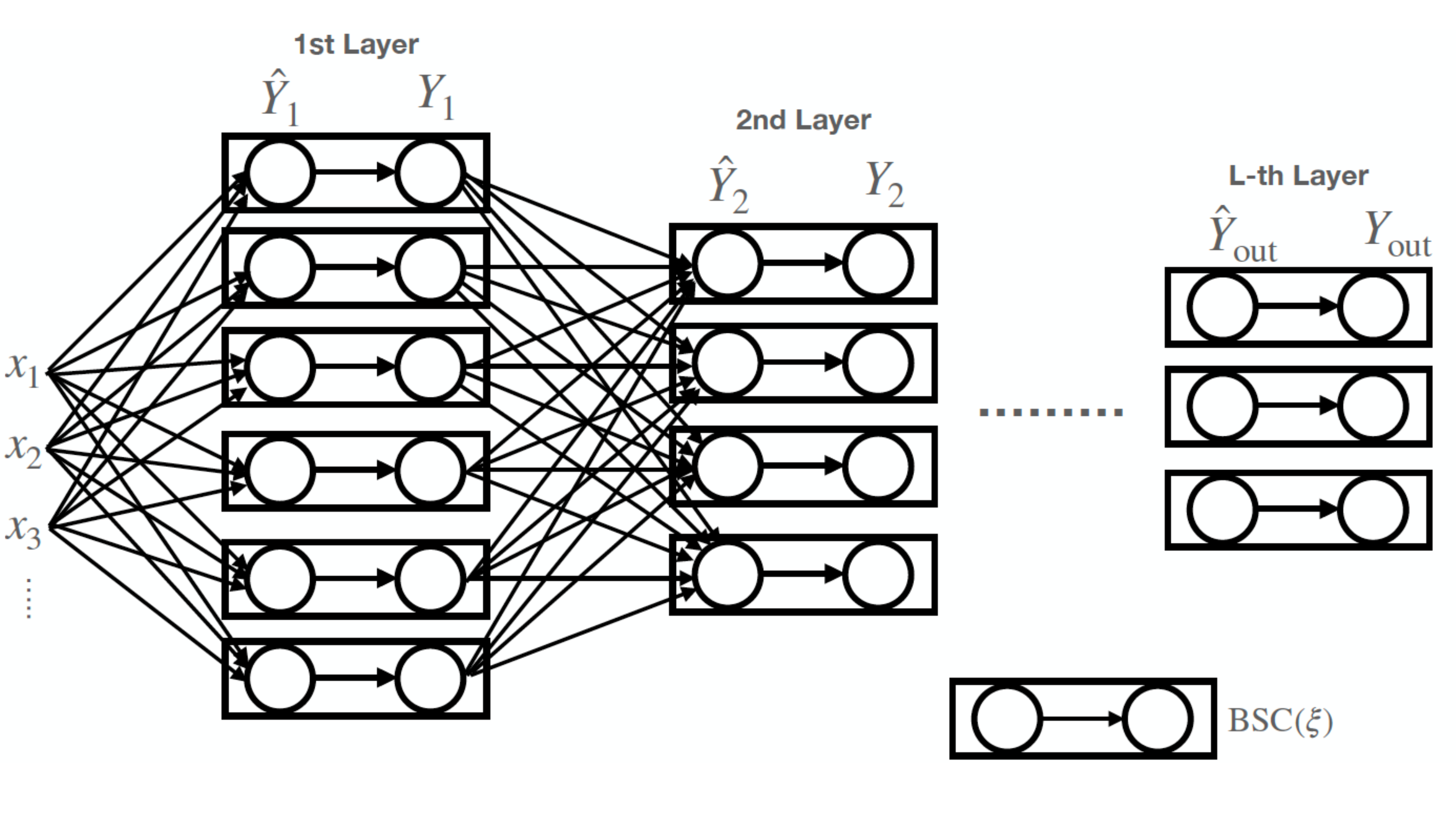}
\caption{Schematic of a noisy neural network with $\xi$-noisy binary neurons. }\label{fig:NoisyNN_schematic}
\end{center}
\end{figure}

\begin{theorem}
\label{thm:2}
\textbf{Information decay in a noisy neural network}
\newline
In a feed-forward simply layered fully-connected binary neural network composed of $L$ layers of $\xi$-noisy neurons, let $X$ be any subset of input to the neural network and $Y_\mathrm{out}$ be the output, then
\begin{equation}
I(X;Y_\mathrm{out})\le \prod_{\ell=1}^{L}\left[1-(4\xi-4\xi^2)^{n_{\ell}}\right]H(X),
\end{equation}
\end{theorem}
where $n_\ell$ is the number of neurons in layer $\ell$. 
\begin{proof}
We prove the theorem by a mathematical induction on the number of layers $L$. Figure \ref{fig:NoisyNN_schematic} gives a schematic of a noisy neural network, let $Y_\ell$ denote the output activations of noisy neurons at layer $\ell$ and $\hat{Y}_\ell$ denote its pre-noise values. We have therefore a Markov chain connecting the input to the output of the network: $X\rightarrow \hat{Y}_1\rightarrow Y_1 \rightarrow \hat{Y}_2\rightarrow Y_2 \rightarrow ... \rightarrow \hat{Y}_\mathrm{L-1}\rightarrow Y_\mathrm{L-1}\rightarrow \hat{Y}_\mathrm{out}\rightarrow Y_\mathrm{out}$. When $L=1$, applying Proposition \ref{prop:1} gives $I(X;Y_\mathrm{out})\le [1-(4\xi-4\xi^2)^{n_1}]I(X;\hat{Y}_\mathrm{out})$ and by definition $I(X;\hat{Y}_\mathrm{out})\le H(X)$. Hence we have $I(X;Y_\mathrm{out})\le [1-(4\xi-4\xi^2)^{n_1}]H(X)$. Now suppose that the claim is true up to a layer number of $K$ with $K\ge 1$, then for $L=K+1$, we have $I(X;Y_\mathrm{out})\le [1-(4\xi-4\xi^2)^{n_{K+1}}]I(X;\hat{Y}_\mathrm{out})$. Since $X \rightarrow Y_K \rightarrow \hat{Y}_\mathrm{out}$, data processing inequality gives $I(X;\hat{Y}_\mathrm{out}) \le I(X;Y_K)$. Treating the first $K$ layers as a new neural network and applying the inductive hypothesis give $I(X;Y_K)\le \prod_{\ell=1}^{K}\left[1-(4\xi-4\xi^2)^{n_{\ell}}\right]H(X)$, therefore $I(X;Y_\mathrm{out})\le \prod_{\ell=1}^{K+1}\left[1-(4\xi-4\xi^2)^{n_{\ell}}\right]H(X)$ and we conclude the proof by induction.
\end{proof}

Theorem \ref{thm:2} relates the rate of information contraction in a noisy neural network to its layer width and overall depth. It implies that wider and shallower noisy neural networks are less susceptible to the decay in information, which rigorizes the intuition that a wider width provides more redundant paths for information to flow through and a deeper depth results in more irrevocable loss. More importantly, Theorem \ref{thm:2} gives an upper bound for such loss that is tighter than previous results. Applying Evans-Schulman method leads to a bound which is the sum over all possible paths from $X$ to the output and Theorem \ref{thm:2} is considerably tighter. In the more special cases where the connections are not dense in the hidden layers on paths from $X$ to the output, the upper bound in Theorem \ref{thm:2} can be further tightened by only requiring the number of neurons in each layer that are connected to $X$.

Using a noisy neural network to compute functions becomes more challenging with a rising level of noise as a result of this information contraction, and there is a point beyond which reliable computation is no longer possible. Fano's inequality gives a lower bound on the minimum mutual information required for reliable computation. The intuition is that if the output is a function of the input with a high probability then their mutual information must be high as well. The following corollary summarizes this result.
\begin{corollary}\label{coro:1}
If a feed-forward simply layered binary neural network with $L$ layers of $\xi$-noisy neurons computes a non-constant function $\delta$-reliably (that is, it produces the correct output for any network input with probability at least $1-\delta$ where $\delta<1/2$). Then
\begin{equation}\label{eqn:44}
\prod_{\ell=1}^{L-1}\left[1-(4\xi-4\xi^2)^{n_{\ell}}\right]\left[1-(4\xi-4\xi^2)\right] \ge \Delta,
\end{equation}
with
\begin{equation}\label{eqn:45}
\Delta = 1+\delta \log_2 \delta + (1-\delta) \log_2 (1-\delta).
\end{equation}
\end{corollary}
\begin{proof}
The neural network implements a non-trivial function of its input, therefore at least one of its output neurons outputs a non-constant Boolean function $f$ of the network input, now we can restrict the proof to this particular output neuron. Obviously the computation of $f$ needs to be at least $\delta$-reliable. $f$ depends on at least one of the input to be non-constant. Let $x_i$ denote an input that $f$ depends on, there exists a setting of the other $n-1$ input variables $\left\{c_k\right\}_{k\neq i}$ such that $f(c_1,c_2,...,c_{i-1},x_i,c_{i+1},...,c_n)$ is either $x_i$ or $\bar{x}_i$. Let $y$ be the random variable measured at the output of that neuron for $x_i$ being a uniformly distributed binary random variable and other inputs to the neural network set to $\left\{c_k\right\}_{k\neq i}$. We apply a decoding function to decode $x_i$ from $y$: $\hat{x}_i=y$ or $\hat{x}_i=\bar{y}$ depending on if $f(c_1,c_2,...,c_{i-1},x_i,c_{i+1},...,c_n)$ is $x_i$ or $\bar{x}_i$. Let $e$ represent occurrence of a decoding error: $\hat{x}_i \neq x_i$. Fano's inequality \cite{cover1999elements} gives $I(x_i;y)=H(x_i)-H(x_i|y)\ge H(x_i) - H(e) - P(e) \log_2(|\mathcal{X}|-1)$, where $P(e)$ is the probability of a decoding error, $H(e)= -P(e)\log_2(P(e))-(1-P(e))\log_2(1-P(e))$ the corresponding binary entropy and $\mathcal{X}=2$ is the support of $x_i$. Recall that $y$ is a $\delta$-reliable representation of $f(c_1,c_2,...,c_{i-1},x_i,c_{i+1},...,c_n)$, hence $P(e)\le \delta$ and  $I(x_i;y)\ge 1+\delta \log_2 \delta + (1-\delta) \log_2 (1-\delta)$. Finally we apply Theorem \ref{thm:2} to the subset of neural network that starts from $x_i$ and ends at $y$ and we have $\prod_{\ell=1}^{L-1}\left[1-(4\xi-4\xi^2)^{n_{\ell}}\right]\left[1-(4\xi-4\xi^2)\right] \ge I(x_i;y)\ge 1+\delta \log_2 \delta + (1-\delta) \log_2 (1-\delta) $.
\end{proof}
Sometimes it is useful to consider the situation where the noisy neural network is used as a feature extractor and the output of the network can then be further processed to calculate a Boolean function $\delta$-reliably. In this case, the inequality becomes $\prod_{\ell=1}^{L}\left[1-(4\xi-4\xi^2)^{n_{\ell}}\right]\ge \Delta$ and no longer limited by a single noisy channel in the output neuron.

With Corollary \ref{coro:1}, we proceed to give a lower bound for the number of noisy neurons for reliable computation. We begin by proving a lemma which will be used in the proof of our next theorem.
\begin{customlemma}{2}\label{lemma:2}
Given $L$ strictly positive integers $n_1$, $n_2$, ..., $n_\ell$ and $0\le a\le 1$, then
\begin{equation}
\prod_{\ell=1}^{L} (1-a^{n_\ell}) \le (1-a^{\frac{\sum_{\ell=1}^{L}n_\ell}{L}})^L.
\end{equation}
\begin{proof}
We use the inequality of arithmetic and geometric means to prove this lemma. First notice that $1-a^{n_\ell}$ is non-negative for $0\le a\le 1$ and we can apply the inequality of arithmetic and geometric means to the product on the left hand side, and we have
\begin{equation}
\prod_{\ell=1}^{L} (1-a^{n_\ell}) \le \left[ \frac{L-\sum_{\ell=1}^L a^{n_\ell}}{L}\right]^L \le (1-a^{\frac{\sum_{\ell=1}^{L}n_\ell}{L}})^L,
\end{equation}
where we have applied the inequality a second time for the sum $\sum_{\ell=1}^L a^{n_\ell} / L$. Equality is attained when all the $n_\ell$ are equal: $n_1=n_2=...=n_\ell$.
\end{proof}
\end{customlemma}
Lemma \ref{lemma:2} can be used to bound the information contraction ratio in a noisy neural network. Particularly, it implies that the contraction upper bound is highest when the neurons are equally partitioned among all hidden layers. More importantly, it connects the information contraction ratio with the total number of neurons in the network, which leads to the following theorem
\begin{theorem} \label{thm:3}
If a feed-forward simply layered binary neural network with $L$ layers of $\xi$-noisy neurons computes a non-constant function $\delta$-reliably, then its total number of hidden layer noisy neurons  $N$ satisfies
\begin{equation}
N \ge N_s(\xi,\delta,L),
\end{equation}
with
\begin{equation} \label{eqn:N_s}
N_s(\xi,\delta,L) = \displaystyle \frac{(L-1)\log\left(1-\left(\displaystyle\frac{\Delta}{1-(4\xi-4\xi^2)}\right)^{1/(L-1)}\right)}{\log(4\xi-4\xi^2)}.
\end{equation}
\end{theorem}
\begin{proof}
In Corollary \ref{coro:1}, applying Lemma \ref{lemma:2} to the left hand side of Inequality \eqref{eqn:44} gives
\begin{equation}
\prod_{\ell=1}^{L-1}\left[1-a^{n_{\ell}}\right](1-a)\le 
\left[1-a^{\sum_{\ell=1}^{L-1} n_{\ell}/(L-1)}\right]^{L-1}(1-a),
\end{equation}
with
\begin{equation}
0\le a=4\xi-4\xi^2\le 1.
\end{equation}
Corollary \ref{coro:1} gives
\begin{equation}
\left(1-a^{\frac{N}{L-1}}\right)^{L-1}(1-a)\ge \Delta.
\end{equation}
Rearranging terms and taking the logarithmic give
\begin{equation}
\log\left(1-\left(\frac{\Delta}{1-a}\right)^{1/(L-1)}\right) \ge \frac{N}{L-1} \log a.
\end{equation}
Since $\log a\le 0$, therefore
\begin{equation}
N \ge \frac{(L-1)\log\left(1-\left(\frac{\Delta}{1-a}\right)^{1/(L-1)}\right)}{\log a}.
\end{equation}
\end{proof}
In Figure \ref{fig:Neuron_number}, we show how this lower bound grows with the noise level $\xi$. The minimum number of neurons required grows quickly when $\xi$ approaches $\delta$. Equation \eqref{eqn:N_s} gives that $N_s$ approaches infinity when $\Delta = 1-(4\xi-4\xi^2)$. This is because that information flow is limited by the last noisy neuron which is the output. The cutoff $\xi$ is actually slightly higher than $\delta$ for $\delta < 0.5$. We believe that this is an artifact of taking the upper bounds that cannot be attained in these cases. Another important fact is that at fixed $\delta$ and $\xi$, $N_s$ increases monotonically with the number of layers $L$, which is showcased in Figure \ref{fig:Neuron_number} as well.

\begin{figure}[th]
\begin{center}
\includegraphics[width=1.\linewidth]{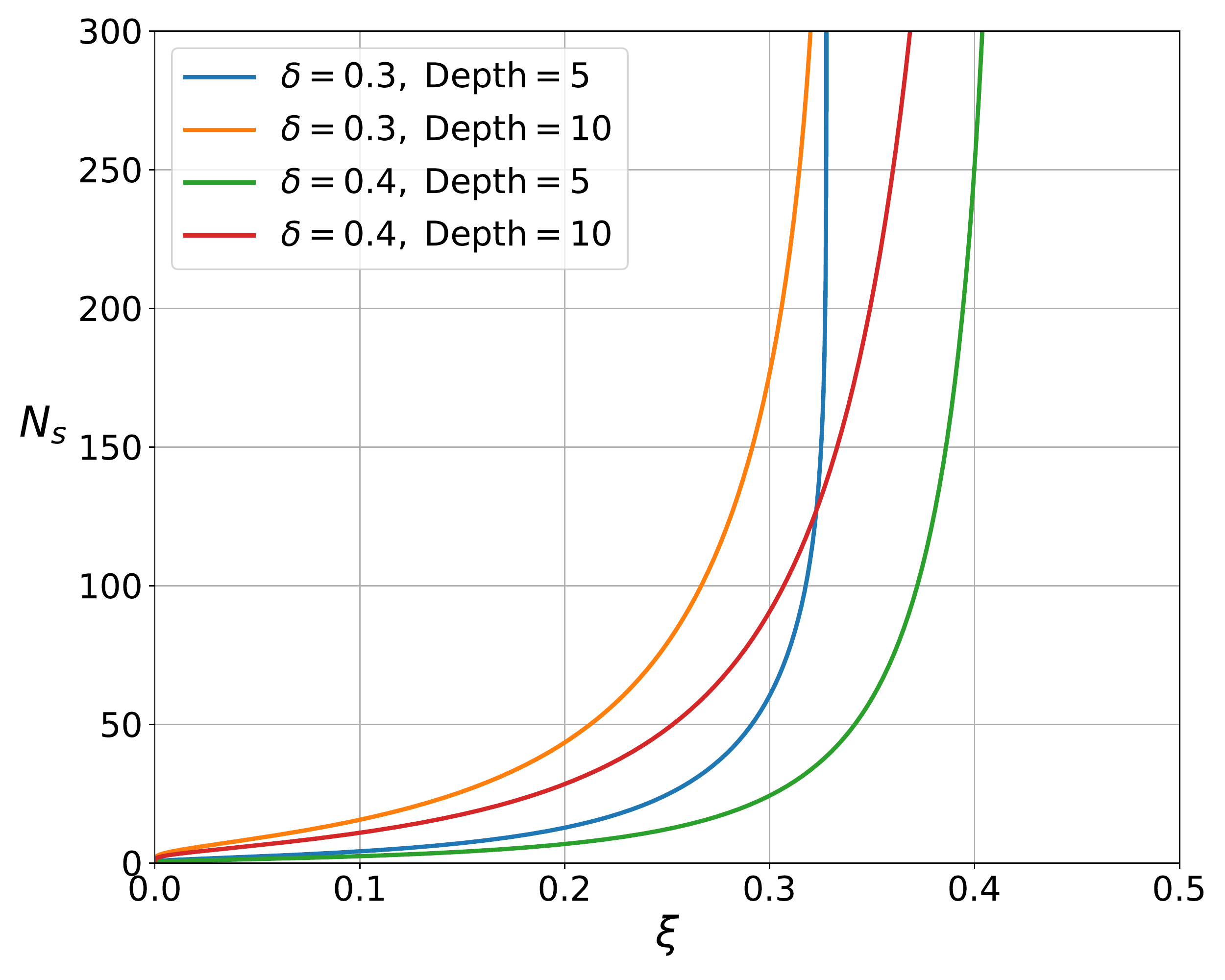}
\caption{Lower bounds for number of neurons $N_s$ in Theorem \ref{thm:3} as a function of the noiseness of neurons evaluated for different computation reliability requirements and network depths.}\label{fig:Neuron_number}
\end{center}
\end{figure}

This particular trade-off between the depth of the noisy network and the number of neurons goes in the opposite direction to various no-flattening results for neural networks and circuits. In terms of expressibility, many depth-width trade-off relationships suggest that flattening or reducing the depth of a network is neuron inefficient. It takes more (sometimes exponentially more) neurons to implement the same function, if limited to arranging them in a smaller number of layers. These results are used as arguments to explain why ``deep" learning works so well \cite{lin2017does}. There are similar results in circuit complexity suggesting that ``depth" helps.
For example, the famous result that $\textit{Parity} \not \in \mathbf{AC}^0$ by Furst \textit{et al.} \cite{furst1984parity} and H\aa stad \cite{haastad1987computational} suggests that implementing the parity function (n element XOR function) over AND, NOT, OR gates of unlimited fan-in costs an exponential number of gates (as a function of input number $n$) if limited to a depth $d$, while there are solutions with linear cost in the number of gates when they are arranged in a tree of depth logarithmic in $n$. Usually these no-flattening results can be expressed as a size complexity function $\Omega(n,d)$ being a lower bound for the number of neurons/gates, which increases with depth $d$. We have the following corollary combining these results with Theorem \ref{thm:3}.


\begin{corollary}\label{coro:2}
Let $f$ be an $n$-input Boolean function of size complexity $\Omega(n,d)$ lower bound in the class of Threshold Circuits (\textbf{TC}, circuits composed of threshold gates or binary neurons with unbounded fan-in and fan-out), suppose that a feed-forward simply layered binary neural network implements function $f$ and is $\delta$-reliable up to a neuron noise level of $\xi$, then at a maximum depth $D$, the minimum number of neurons required in this network is 
\begin{equation}
\min_{1\le d \le D}\max \left(\Omega(n,d), N_s(\xi,\delta,d)+1 \right).
\end{equation}
\end{corollary}
Corollary \ref{coro:2} recognizes that in order for a network to have both enough expressive power and noise robustness, it needs to satisfy both minimum size requirements. The $+1$ accounts for the last output neuron. $N_s$ monotonically increases with the network depth and $\Omega$ monotonically decreases with the network depth. When they cross each other, there is an optimal depth that gives the minimum size complexity requirement. Intuitively, the network cannot be too long and thin since it cannot be noise resistant enough as we have shown; and nor can it be too shallow as the expressibility requirement makes it very neuron inefficient. The function $\Omega(n,d)$ is $f$-dependent and are obtained from circuit complexity theory of the complexity class \textbf{TC}. The complexity lower bound study of class \textbf{TC} is currently a frontier field of research and many efforts are still on-going \cite{chen2019bootstrapping}. We give an example for $f$ being the parity function defined as 
\begin{equation}
f(x) = x_1 \oplus x_2 \oplus ... \oplus x_n,
\end{equation}
where $\oplus$ denotes XOR. Impagliazzo \textit{et al.} \cite{impagliazzo1997size} proved that any threshold circuit of depth $d$ that computes the parity function has at least $(n/2)^{1/2(d-1)}$ threshold gates or binary neurons in our terminology. Therefore, we can have $\Omega(n,d)=(n/2)^{1/2(d-1)}$, which holds as the absolute lower bound on the number of neurons. It may be possible to further improve it, since we restrict the network structure to be simply layered. 

\begin{figure}[t]
\begin{center}
\includegraphics[width=1.\linewidth]{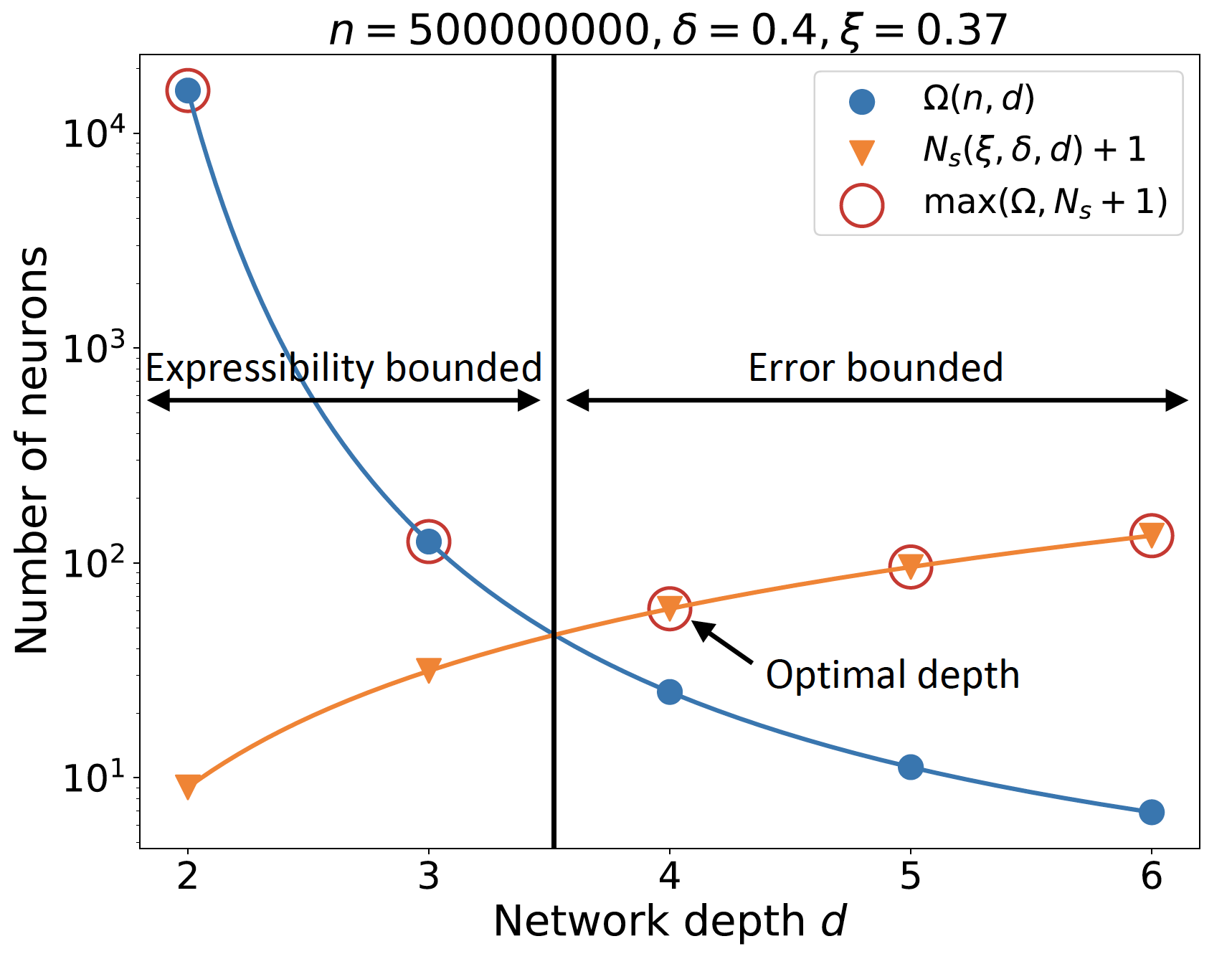}
\caption{Minimum number of neurons required according to Corollary \ref{coro:2} for the parity function, blue dots are the size complexity bounds for the parity function $\Omega(n,d)=(n/2)^{1/2(d-1)}$ and orange triangles are the noise robustness bounds. The minimum number of neurons required needs to be more than the maximum of both bounds, which are shown by red circles.} \label{fig:two_regimes}
\end{center}
\end{figure}

Figure \ref{fig:two_regimes} shows a case for the parity function with $n=5\times 10^8$, $\delta=0.4$, $\xi=0.37$. $\Omega(n,d)$ (blue curve) monotonically decreases with $d$ and $N_s+1$ (orange curve) monotonically increases with $d$. The crossing happens between $d=3$ and $d=4$. To the left of this crossing point, the minimum number of neurons required is limited by the required expressive power of the network and to the right of this crossing point it is limited by the noise robustness requirement. There are therefore two distinct regimes where one is expressibility bounded and the other is error bounded. The optimal number of layers that achieves the minimum neuron number requirement in this case is $4$ with a minimum lower bound of $61.22$ neurons. (Note that as we are obtaining results from bounds that are not necessarily tight, the actual optimal number might be different.) This is almost a ten-fold increase from the minimum requirement in the absence of noise for $d\le 6$ which is $\Omega(n,6)=6.915$. Performing computation with noisy neurons while retaining noise-robustness requires using more neurons. Our paper gives a lower bound for such increase. 

The results we have derived in this section can be extended to the case of weakly-correlated layer-wise neuron noise using Proposition \ref{prop:2}. 
Qualitatively speaking, Corollary \ref{coro:1bis} indicates that weakly-correlated neuron noise results in an end-to-end information contraction upper bound that is lower than the independent noise case at equal $\xi$, thus requiring a further increase in the minimum number of neurons to combat loss of information.

One way to understand this increase of the number of neurons is via considering error correction mechanism. 
As the gates becomes noisy, each logical bit of information would require multiple physical bits to encode so that the robustness condition can be satisfied, which is similar to the error-correction overhead that we will discuss the next.

\section{Application to fault-tolerant cellular automata}
\label{Sec:CA}

Since the early works of Turing~\cite{turing1990chemical}, people are interested in the ``self-organization'' phenomena, where stable patterns form in a dynamical system, despite dissipation and noise. In the language of statistical physics, these stable patterns can be regarded as phases of the system at a finite temperature~\cite{liggett2012interacting}; while to computer scientists, the fact that different initial conditions can lead to unique different patterns means that information can be reliably stored in these complex systems, which serve as fault-tolerant memory. A simplified model for studying these phenomena, adopted by both statistical physics~\cite{liggett2012interacting,wolfram1983statistical,grinstein1985statistical} and computer scientists~\cite{gacs2001reliable,gray2001reader,toom1974nonergodic}, is the cellular automata~\cite{wolfram1983statistical}, where time and space are both discrete, similar to the binary neurons being considered in this paper. In fact, similar to Von Neumann's circuit approach, fault-tolerant cellular automata form another class of model capable of reliable universal computation with unreliable components~\cite{gacs2001reliable}. More recently, these systems are also found to be connected to a novel phase of matter called time-crystals~\cite{yao2020classical}.

\begin{figure}[t]
\begin{center}
\includegraphics[width=1.\linewidth]{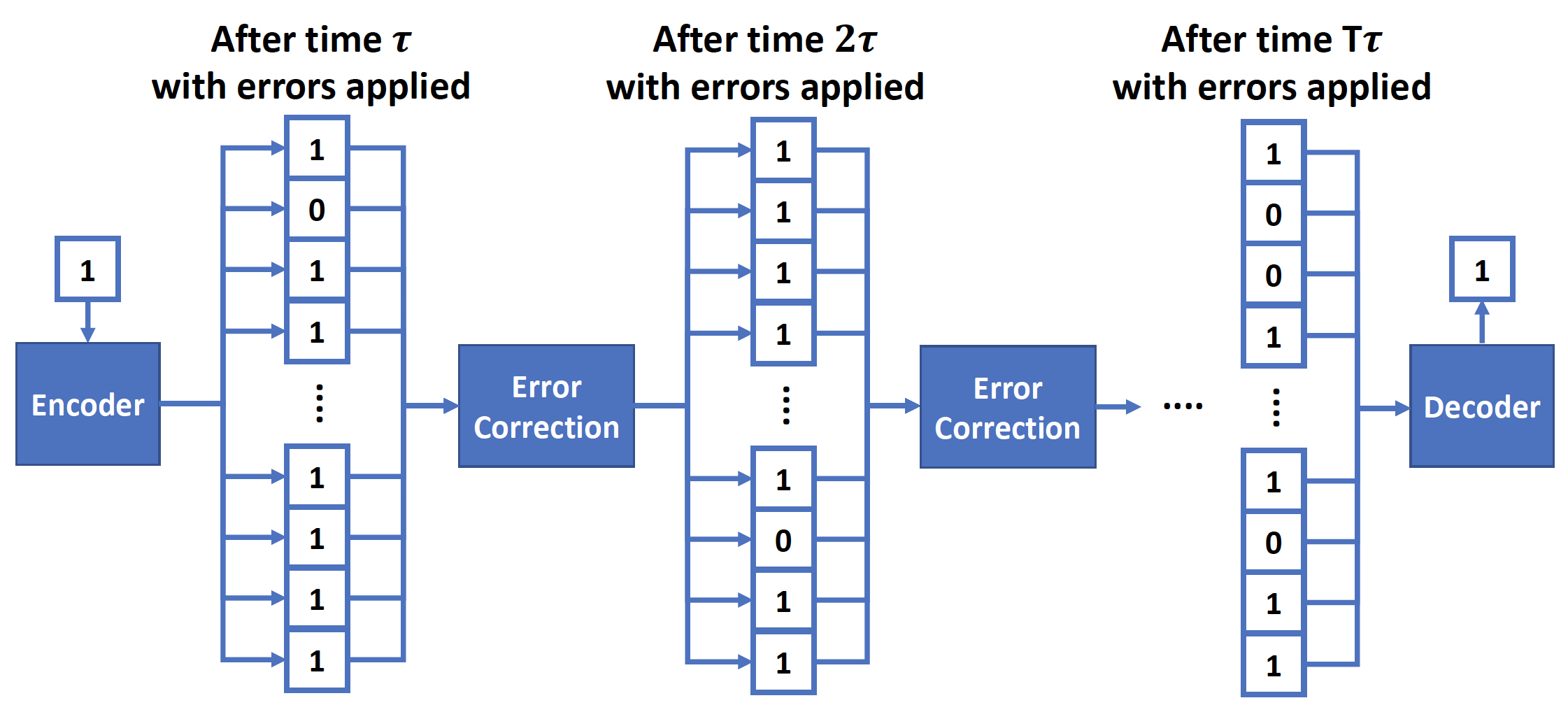}
\caption{Schematic of fault-tolerant memory with error correction, the system is unrolled in time dimension, it corresponds to the same physical memory array. $\tau$ is the time between two consecutive memory ``refreshes". One logical bit is protected over a long period of time despite continuous corruption by random noise that happens to physical memory bits.} \label{fig:error_correction}
\end{center}
\end{figure}

In this section, we apply our information contraction results to fault-tolerant cellular automata with any (global or local) update rules. For simplicity, we consider the capability of these systems serving as fault-tolerant memory. 
Without loss of generality, we assume that such systems encode 1 bit of information with $n$ noisy physical memory bits and apply some error correction rules at constant intervals of time (a time step of $\tau$) only using information stored in noisy physical memory bits, meanwhile errors may occur randomly in these memory cells. Our approach may apply to wider classes of systems, however, in this paper we will consider fault-tolerant cellular automata~\cite{gacs2001reliable} as an example of fault-tolerant memory.

An important question to ask is: given $\xi$-noisy memory cells, what is the minimum memory overhead $n(\delta, T,\xi)$ required such that 1 bit of logical information can be $\delta$-reliably retrieved (correctly decoded with probability at least $1-\delta$) after $T$ intervals in time? This time duration $T(\delta,n,\xi)$ (normalized by $\tau$) with $\delta$ being a constant (e.g. $2/3$ as in \cite{gacs2001reliable}) is also called the relaxation time of the system. While any finite system has a finite relaxation time, stable phases in statistical physics are defined as the case of $T(\delta,n,\xi)\to\infty$ at the thermodynamic limit of $n\to\infty$. In this sense, obtaining bounds of the relaxation time is crucial for proving the stability of phases of matter. It turns out that our information contraction results give generic bounds for these important quantities applicable to any possible error correction rules---we obtain a generic lower bound of the error correction overhead $n$ or equivalently an upper bound for the relaxation time $T$.

In Figure \ref{fig:error_correction}, we show a schematic of how such a system works. It bears a strong resemblance to the noisy binary neural networks we have studied so far in previous sections. The number of time intervals $T$ is analogous to the number of layers $L$, random errors in memory can be considered to be introduced by independent binary noisy channels and the error correction block can be considered as a generic data processing function similar to the function realized by a layer of neural network. This system is analogous to a noisy binary neural network of constant width $n$ and of depth $T$. It is also simply layered since it is impossible to have skip connections in time dimension.

Therefore, if errors occur uniformly and independently with probability $\xi$ during a time period of $\tau$ between two error correction points, the results of Theorem \ref{thm:2} and Corollary \ref{coro:1} still hold with slight modifications. In order to decode the original logical bit at least $\delta$-reliably after $T$ intervals of time, the residual mutual information in the noisy physical memory about the original bit needs to satisfy
\begin{equation}
\left[1-(4\xi-4\xi^2)^{n}\right]^T \ge \Delta(\delta),
\end{equation}
with $\Delta$ defined as in Equation \eqref{eqn:45}. This result is no more limited by the single output noisy channel since we consider the decoding function to be a noiseless process. Hence we obtain a lower bound for $n$:
\begin{equation}
n \ge \frac{\log(1-\Delta^{1/T}(\delta))}{\log(4\xi -4\xi^2)}.
\end{equation}
This is another impossibility result derived from our information theoretic approach, it says that for given $\delta$, $\xi$ and $T$, it is impossible to use less overhead than what is given by the lower bound above. 
Figure \ref{fig:ECLB} shows numerical values of this lower bound plotted as a function of $T$ for relevant values of $\delta$ and $\xi$ for fault-tolerant memory.
\begin{figure}[t]
\begin{center}
\includegraphics[width=1.\linewidth]{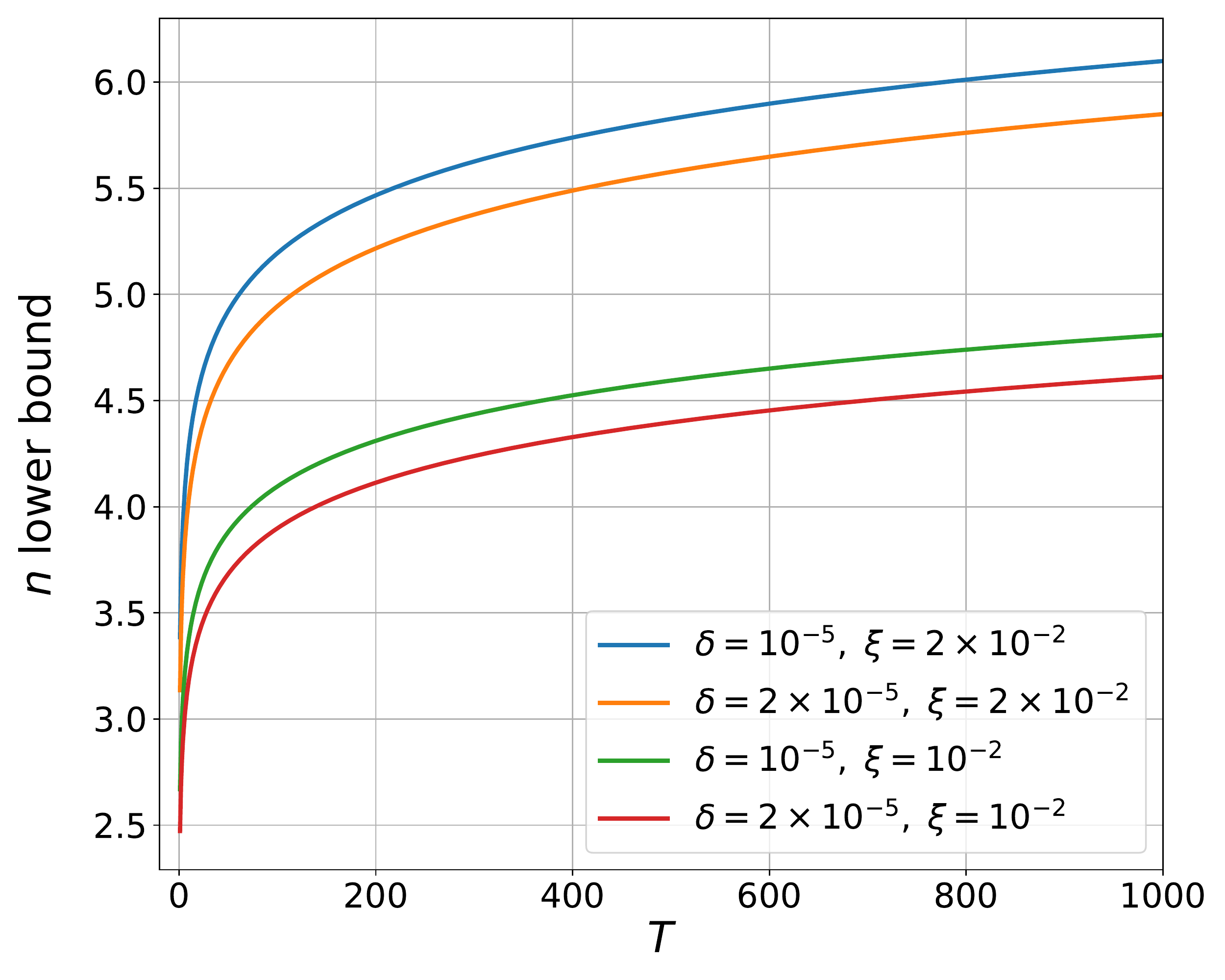}
\caption{The error correction overhead lower bound plotted for two different values of $\delta$ and $\xi$ as a function of the number of time intervals $T$.} \label{fig:ECLB}
\end{center}
\end{figure}
Alternatively, this result can be stated as an upper bound of the relaxation time $T$
\begin{equation}
T \le \frac{\log \Delta(\delta)}{\log(1-(4\xi -4\xi^2)^n)}.
\end{equation}
At fixed $\xi$ and $\delta$, when $n$ goes to infinity, this upper bound can be approximated using a Taylor series expansion and is asymptotically equivalent to $\log \Delta^{-1} \exp(n\log(1/(4\xi -4\xi^2)))$. 
Relaxation time upper bound grows exponentially with system size $n$, at sufficiently big $n$, there exists $C_1(\delta)>0$ such that 
\begin{equation}
T \le C_1(\delta)\exp\left(n \log \left(\frac{1}{4\xi - 4\xi^2}\right) \right).
\label{TUB}
\end{equation}

We derive a lower bound for $T$ in the case of simple repetition encoding and error correction using global majority vote. 
In this case, an error cannot be corrected only if a catastrophic event happens where more than half of the physical memory bits are flipped by random noise before the next error correction happens, we let $p_e$ denote this probability. After $T\tau$, the information can be correctly decoded with a global majority vote if catastrophic events have happened an even number of times between time step $1$ and $T$. Since these events are independent and their distribution follows a binomial distribution. Using the relationship between odd and even binomial terms, it is easy to get the probability of correct decoding as $(1+(1-2p_e)^T)/2$. Therefore $(1+(1-2p_e)^T)/2=1-\delta$ and 
\begin{equation}
T=\frac{\log(1-2\delta)}{\log(1-2p_e)}.
\end{equation}
Probability $p_e$ can be calculated as the probability of tail events in Bernoulli trials, which can be bounded with the Chernoff-Hoeffding bound (the more general form of the Chernoff bounds involving relative entropy, see Theorem 1 in \cite{hoeffding1994probability}) as 
\begin{align}
p_e &\le \exp\left( -\frac{n}{2} \log\left(\frac{1}{2\xi}\right)-\frac{n}{2}\log\left(\frac{1}{2(1-\xi)}\right)\right) \nonumber\\ 
&= \exp\left( -\frac{n}{2} \log \left(\frac{1}{4\xi - 4\xi^2}\right) \right).
\end{align}
Hence for sufficiently big $n$, there exists $C_2(\delta)>0$ such that
\begin{equation}
T \ge C_2(\delta)\exp\left(\frac{n}{2} \log \left(\frac{1}{4\xi - 4\xi^2}\right) \right).
\label{TLB}
\end{equation}
Comparing our upper bound in Inequality~\eqref{TUB} and the lower bound in Inequality~\eqref{TLB}, our information contraction results give a nontrivial upper bound of the relaxation time $T$ that is asymptotically tight up to a constant factor of $2$ in the exponent. 
Relaxation time that grows exponentially with size $n$ of the system is provably the best that can be achieved. Global majority vote is expensive to implement, there are local update rules such as the 2-dimensional Toom's rule \cite{toom1974nonergodic} that can also achieve exponential growth in relaxation time \cite{berman1988investigations}.

We also note that some fault-tolerant cellular automata can be stable even if the errors are correlated~\cite{gacs2001reliable}, and results similar to Proposition~\ref{prop:2} can potentially be developed for fault-tolerant cellular automata.

\section{Conclusion}
\label{Sec:conclusion}
In this paper, we have made contributions to the fundamental understanding of neural networks or general networks with noisy components. Information theoretic tools are shown to be essential to this study. We prove a strong data processing inequality which may find broader application and is thus interesting in its own right. The information contraction bounds given by our method improve the previous results by Evans and Schulman. Information contraction bounds are derived for both independent and weakly-correlated neuron noise models. Leveraging these results, we have shown that information contraction in noisy neural networks requires an increase in the minimum number of neurons to provide adequate noise robustness. Robust information flow favors wide and shallow networks over narrower and deeper ones, while network expressibility requirements favor the opposite. An optimal minimum network depth which gives the most neuron efficient lower bound is obtained from the balance of these two effects. Additionally, we show an application of our results to fault-tolerant cellular automata, showcasing wider applicability of our approach. This work can be readily applied to advance the understanding of information processing systems found in neuroscience, physics and electronics, where noise is naturally present.

Finally, we list some potential future directions. We have derived an SDPI for discrete channels and therefore considered binary neural networks; hence, the generalization of our SDPI results to the continuous-variable case is crucial to precise information bounds for general neural networks beyond the binary case. 
We have considered independent neuron noise and a special case of weakly-correlated neuron noise, more general correlated neuron noise models in noisy neural networks can potentially be studied by our Theorem~\ref{theo:1}. 
It may also be interesting to use Theorem~\ref{theo:1} to derive network weight dependent SDPIs and rule-dependent relaxation time for different fault-tolerant cellular automata.

\textit{Note added.} After the completion of this work, we have become aware of a recent extension~\cite{ordentlich2020strong} of Theorem 21 of Ref.~\cite{polyanskiy2017strong} that can lead to our Theorem \ref{theo:1} via a different proof.


%
\appendices

\section{Proof of Lemma~\ref{lemma:1}} 
\label{App:lemma1}
Let $M$ denote the supremum:
\begin{equation}
M=\sup_{\substack{\forall \boldsymbol{c}\in {\rm I\!R}^{n}\\
\forall \boldsymbol{x}\in \mathring{\mathcal{P}}_n}}
\frac{\boldsymbol{c}^{\intercal}\mathbf{H}_{f}(\boldsymbol{x})\boldsymbol{c}}{\boldsymbol{c}^{\intercal}\mathbf{H}_{g}(\boldsymbol{x})\boldsymbol{c}}.
\end{equation}
If $M$ is finite, by definition, we have $\forall \boldsymbol{c}\in {\rm I\!R}^{n}$ and $\forall \boldsymbol{x}\in \mathring{\mathcal{P}}_n$, $\boldsymbol{c}^{\intercal}[\mathbf{H}_{f}(\boldsymbol{x})-M\mathbf{H}_{g}(\boldsymbol{x})]\boldsymbol{c}\le 0$. Additionally, linearity of taking the second derivatives gives $\mathbf{H}_{f}(\boldsymbol{x})-M\mathbf{H}_{g}(\boldsymbol{x})= \mathbf{H}_{f-Mg}(\boldsymbol{x})$.
The Hessian matrix $\mathbf{H}_{f-Mg}(\boldsymbol{x})$ is negative semi-definite.
Hence function $f-Mg$ is also a concave function on $\mathring{\mathcal{P}}_n$.
Apply Jensen's inequality to the concave function $f-Mg$ for $d$ points $(\boldsymbol{x}_1,...,\boldsymbol{x}_d)$ on $\mathring{\mathcal{P}}_n$ with weights $(p_1,...,p_d)$
\begin{equation}
\sum_{i=1}^{d}p_i(f-Mg)(\boldsymbol{x}_i)\le(f-Mg)\left(\sum_{i=1}^d p_i\boldsymbol{x}_i\right).
\end{equation}
Rearrange the terms, we have for $\forall (\boldsymbol{x}_1,...,\boldsymbol{x}_d)\in \mathring{\mathcal{P}}_{n}^d,
\forall(p_1,...,p_d)\in\mathcal{P}_{d-1}$, 
\begin{equation}
\begin{aligned}
&\sum_{i=1}^{m}p_i f(\boldsymbol{x}_i)-f(\sum_{i=1}^m p_i\boldsymbol{x}_i) \le\\
&M\left[ \sum_{i=1}^{m}p_i g(\boldsymbol{x}_i)-g(\sum_{i=1}^m p_i\boldsymbol{x}_i)\right].
\end{aligned}
\end{equation}
Therefore,
\begin{equation}
\sup_{\substack{\forall (\boldsymbol{x}_1,...,\boldsymbol{x}_d)\in \mathring{\mathcal{P}}_{n}^d\\
\forall(p_1,...,p_d)\in\mathcal{P}_{d-1}}}\frac{f_2(\boldsymbol{x}_1,...,\boldsymbol{x}_d;p_1,...,p_d)}{g_2(\boldsymbol{x}_1,...,\boldsymbol{x}_d;p_1,...,p_d)}
\le M.
\end{equation}
The above inequality trivially holds when $M$ is infinite. We proceed to show that the upper bound $M$ can be effectively attained by the supremum on the left-hand side for the equality to hold. There exists a sequence $(\tilde{\boldsymbol{x}}_j)$ in $\mathring{\mathcal{P}}_n$ such that $\exists \tilde{\boldsymbol{c}} \in  {\rm I\!R}^{n}$ and $\left \|\tilde{\boldsymbol{c}}  \right \|_2=1$,  $\lim_{j\to\infty}\frac{\tilde{\boldsymbol{c}}^{\intercal}\mathbf{H}_{f}(\tilde{\boldsymbol{x}}_j)\tilde{\boldsymbol{c}}}{\tilde{\boldsymbol{c}}^{\intercal}\mathbf{H}_{g}(\tilde{\boldsymbol{x}}_j)\tilde{\boldsymbol{c}}} = M$.
Consider $\left \{\boldsymbol{x}_1,...,\boldsymbol{x}_d\right\}$ and $\boldsymbol{x}_i=\tilde{\boldsymbol{x}}_j + \Delta \boldsymbol{x}_{i}$ with $j$ fixed, in the limit where all $\Delta \boldsymbol{x}_i$ vanish: $\Delta\boldsymbol{x}_i\to 0$, we can expand $f_2$ and $g_2$ in their Taylor series up to the second order in $\Delta \boldsymbol{x}_i$
\begin{align}
&f_2(\boldsymbol{x}_1,...,\boldsymbol{x}_d;p_1,...,p_d)
= \frac{1}{2} \sum_{i=1}^d p_i \Delta\boldsymbol{x}_i^{\intercal}
\mathbf{H}_f(\tilde{\boldsymbol{x}}_j)
\Delta\boldsymbol{x}_i\nonumber\\ 
& -\frac{1}{2} (\sum_{i=1}^{d} p_i \Delta\boldsymbol{x}_i)^{\intercal}
\mathbf{H}_f(\tilde{\boldsymbol{x}}_j)
(\sum_{i=1}^{d} p_i \Delta\boldsymbol{x}_i)
+ \mathcal{O}(\left\|\Delta\boldsymbol{x}_i\right\|^3), \\
&g_2(\boldsymbol{x}_1,...,\boldsymbol{x}_d;p_1,...,p_d)
= \frac{1}{2} \sum_{i=1}^d p_i \Delta\boldsymbol{x}_i^{\intercal}
\mathbf{H}_g(\tilde{\boldsymbol{x}}_j)
\Delta\boldsymbol{x}_i  \nonumber\\
&-
\frac{1}{2} (\sum_{i=1}^{d} p_i \Delta\boldsymbol{x}_i)^{\intercal}
\mathbf{H}_g(\tilde{\boldsymbol{x}}_j)
(\sum_{i=1}^{d} p_i \Delta\boldsymbol{x}_i)
+ \mathcal{O}(\left\|\Delta\boldsymbol{x}_i\right\|^3).
\end{align}
It is easy to verify that the zeroth order and first order terms all vanish in these expansions. Now choose $\Delta \boldsymbol{x}_i=\lambda_i \tilde{\boldsymbol{c}}$ (this is obviously possible with small enough $\lambda_i$ since $\mathring{\mathcal{P}}_{n}$ is an open set, geometrically we are restricting all $\boldsymbol{x}_i$ to be on the line defined by $\tilde{\boldsymbol{x}}_j + \lambda \tilde{\boldsymbol{c}}$, $\lambda \in {\rm I\!R}$), then
\begin{align}
&f_2(\boldsymbol{x}_1,...,\boldsymbol{x}_d;p_1,...,p_d)
= \nonumber\\
& \frac{1}{2} \left[\sum_{i=1}^d p_i\lambda_i^2
-\left(\sum_{i=1}^{d}p_i\lambda_i\right)^2\right] \tilde{\boldsymbol{c}}^{\intercal}
\mathbf{H}_f(\tilde{\boldsymbol{x}}_j)
\tilde{\boldsymbol{c}}
+ \mathcal{O}(\lambda_i^3), \\
&g_2(\boldsymbol{x}_1,...,\boldsymbol{x}_d;p_1,...,p_d)
= \nonumber\\ 
&\frac{1}{2} \left[\sum_{i=1}^d p_i\lambda_i^2
-\left(\sum_{i=1}^{d}p_i\lambda_i\right)^2\right]\tilde{ \boldsymbol{c}}^{\intercal}
\mathbf{H}_g(\tilde{\boldsymbol{x}}_j)
\tilde{\boldsymbol{c}}
+ \mathcal{O}(\lambda_i^3).
\end{align}
Strict convexity of the square function $x\rightarrow x^2$ gives $\sum_{i=1}^d p_i \lambda_i^2 \ge (\sum_{i=1}^d p_i\lambda_i)^2$, the identity is attained if and only if all the $\lambda_i$ are equal or $(p_1, ...,p_d)$ being a one-hot weight vector. It suffices to choose the values of $\lambda_i$ that are different from each other and $p_i = 1/d$ for example, in the limit where $\forall i$, $\lambda_i \to 0$. The ratio $f_2/g_2$ tends to ${\tilde{\boldsymbol{c}}^{\intercal}\mathbf{H}_{f}(\tilde{\boldsymbol{x}}_j)\tilde{\boldsymbol{c}}}/{\tilde{\boldsymbol{c}}^{\intercal}\mathbf{H}_{g}(\tilde{\boldsymbol{x}}_j)\tilde{\boldsymbol{c}}}$.
Now, let $j$ go to infinity, since the two limiting operations are independent from each other, we have a sequence of $(\boldsymbol{x}_1,...,\boldsymbol{x}_d)\in \mathring{\mathcal{P}}_{n}^d$ and $(p_1,...,p_d)\in\mathcal{P}_{d-1}$ such that the ratio $f_2/g_2$ tends to $M$. We have previously proven that $M$ is also an upper bound, hence
\begin{equation}
\sup_{\substack{\forall (\boldsymbol{x}_1,...,\boldsymbol{x}_d)\in \mathring{\mathcal{P}}_{n}^d\\
\forall(p_1,...,p_d)\in\mathcal{P}_{d-1}}}
\frac{f_2(\boldsymbol{x}_1,...,\boldsymbol{x}_d;p_1,...,p_d)}{g_2(\boldsymbol{x}_1,...,\boldsymbol{x}_d;p_1,...,p_d)}
= M.
\end{equation}
This completes the proof of Lemma \ref{lemma:1}.

\section{More details of proof of Theorem \ref{theo:1}}
\label{App:theorem1}
\begin{flushleft} 
In this Appendix we give additional details that were omitted in the proof of Theorem \ref{theo:1}. 
The matrix form of $\mathbf{H}_{f}(\boldsymbol{p})$ is
\end{flushleft} 
\begin{strip}
\begin{equation}
\mathbf{H}_{f}(\boldsymbol{p})
=-
\begin{pmatrix}
\displaystyle\sum_{j=1}^m\frac{(a_{1,j}-a_{n,j})^2}{\sum_{i=1}^n p_i a_{i,j}}  
&\displaystyle\sum_{j=1}^m\frac{(a_{1,j}-a_{n,j})(a_{2,j}-a_{n,j})}{\sum_{i=1}^n p_i a_{i,j}}
&  \dots   
&\displaystyle\sum_{j=1}^m\frac{(a_{1,j}-a_{n,j})(a_{n-1,j}-a_{n,j})}{\sum_{i=1}^n p_i a_{i,j}}  \\
\displaystyle\sum_{j=1}^m\frac{(a_{2,j}-a_{n,j})(a_{1,j}-a_{n,j})}{\sum_{i=1}^n p_i a_{i,j}}
& \displaystyle\sum_{j=1}^m\frac{(a_{2,j}-a_{n,j})^2}{\sum_{i=1}^n p_i a_{i,j}} & \ddots   & \vdots \\
\vdots  &  \ddots       & \ddots  & 
\vdots
\\
\displaystyle\sum_{j=1}^m\frac{(a_{n-1,j}-a_{n,j})(a_{1,j}-a_{n,j})}{\sum_{i=1}^n p_i a_{i,j}}
& \dots 
& \dots  
& \displaystyle\sum_{j=1}^m\frac{(a_{n-1,j}-a_{n,j})^2}{\sum_{i=1}^n p_i a_{i,j}} 
\end{pmatrix}.
\end{equation}
\end{strip}
\begin{flushleft} 
In order to obtain Equation \eqref{eqn:Q_g_int}, we used the identity
\end{flushleft}
\begin{strip}
\begin{align}
\sum_{j=1}^m \frac{p_s a_{t,j}+p_{t} a_{s,j}}{p_s+p_{t}}
=\sum_{j=1}^m \frac{p_s p_{t}(a_{s,j}-a_{t,j})^2}{(p_s+p_{t})\sum_{i=1}^{n+1} p_i a_{i,j}}
+\sum_{j=1}^m \left[\frac{(p_s+p_{t})a_{s,j} a_{t,j}}{\sum_{i=1}^{n+1} p_i a_{i,j}}
+\frac{\sum_{u \ne s, u\ne t}p_s p_u a_{t,j}a_{u,j}+p_t p_{u}a_{s,j}a_{u,j}}{(p_s+p_{t})\sum_{i=1}^{n+1} p_i a_{i,j}}\right].
\end{align}
\end{strip}

\begin{flushleft} 
The full expression of $Q_g(\boldsymbol{c})$ is then:
\end{flushleft}
\begin{strip}
\begin{align}
Q_g(\boldsymbol{c})&=\sum_{s=1}^n \left(c_s^2+\sum_{t \ne s}c_s c_t\right)\sum_{j=1}^m \frac{(a_{s,j}-a_{n+1,j})^2}{\sum_{i=1}^{n+1} p_i a_{i,j}}-\sum_{1\le s <t \le n} c_s c_t \sum_{j=1}^m \frac{(a_{s,j}-a_{t,j})^2}{\sum_{i=1}^{n+1} p_i a_{i,j}} \nonumber\\
&\;+\sum_{s=1}^n \left(c_s^2+\sum_{t \ne s}c_s c_t\right)
\sum_{j=1}^m \left[\frac{(p_s+p_{n+1})^2a_{s,j} a_{n+1,j}}{p_s p_{n+1}\sum_{i=1}^{n+1} p_i a_{i,j}}
+\frac{\sum_{u \ne s, u\ne n+1}p_s p_u a_{n+1,j}a_{u,j}+p_{n+1} p_{u}a_{s,j}a_{u,j}}{p_s p_{n+1}\sum_{i=1}^{n+1} p_i a_{i,j}}\right]\nonumber \\
&\;-\sum_{1\le s <t \le n} c_s c_t \sum_{j=1}^m\left[\frac{(p_s+p_{t})^2a_{s,j} a_{t,j}}{p_s p_t \sum_{i=1}^{n+1} p_i a_{i,j}}
+\frac{\sum_{u \ne s, u\ne t}p_s p_u a_{t,j}a_{u,j}+p_t p_{u}a_{s,j}a_{u,j}}{p_s p_t \sum_{i=1}^{n+1} p_i a_{i,j}}\right].
\end{align}
\end{strip}
\begin{flushleft} 
The full expression of $Q_f(\boldsymbol{c})$ is:
\end{flushleft}
\begin{strip}
\begin{align}
Q_f(\boldsymbol{c})&=\sum_{s=1}^n c_s^2 \sum_{j=1}^m \frac{(a_{s,j}-a_{n+1,j})^2}{\sum_{i=1}^{n+1} p_i a_{i,j}}
+\sum_{1\le s < t \le n} c_s c_t \sum_{j=1}^m\frac{2(a_{s,j}-a_{n+1,j})(a_{t,j}-a_{n+1,j})}{\sum_{i=1}^{n+1} p_i a_{i,j}} \nonumber \\
&=\sum_{s=1}^n c_s^2 \sum_{j=1}^m \frac{(a_{s,j}-a_{n+1,j})^2}{\sum_{i=1}^{n+1} p_i a_{i,j}}
+\sum_{1\le s < t \le n} c_s c_t \sum_{j=1}^m\frac{(a_{s,j}-a_{n+1,j})^2+(a_{t,j}-a_{n+1,j})^2-(a_{s,j}-a_{t,j})^2}{\sum_{i=1}^{n+1} p_i a_{i,j}}  \nonumber \\
&=\sum_{s=1}^n \left(c_s^2+\sum_{t \ne s}c_s c_t\right)\sum_{j=1}^m \frac{(a_{s,j}-a_{n+1,j})^2}{\sum_{i=1}^{n+1} p_i a_{i,j}}-\sum_{1\le s <t \le n} c_s c_t \sum_{j=1}^m \frac{(a_{s,j}-a_{t,j})^2}{\sum_{i=1}^{n+1} p_i a_{i,j}}.
\end{align}
\end{strip}

\newpage 
\
\newpage
\clearpage

\begin{flushleft} 
The full derivations of $Q_{s,t}(\boldsymbol{c})$ and $Q_{s,n+1}(\boldsymbol{c})$ are:
\end{flushleft}
\begin{strip}
\begin{align}
Q_{s,t}(\boldsymbol{c})&=
\frac{p_t}{p_s}\left(c_s^2+\sum_{u\ne s} c_s c_u\right)
+\frac{p_s}{p_t}\left(c_t^2+\sum_{u \ne t} c_t c_u\right)
-\left(2+\frac{p_s}{p_t}+\frac{p_t}{p_s}\right)c_s c_t
-\frac{p_s}{p_t} \sum_{u\ne s, u \ne t} c_u c_t
-\frac{p_t}{p_s} \sum_{u\ne s, u \ne t} c_u c_s \nonumber\\ 
&=\frac{p_t}{p_s}c_s^2+\frac{p_s}{p_t}c_t^2-2c_s c_t \nonumber\\
&=\left(\sqrt{\frac{p_t}{p_s}}c_s - \sqrt{\frac{p_s}{p_t}}c_t\right)^2, \\
Q_{s,n+1}(\boldsymbol{c})&=
\left(c_s^2+\sum_{t \ne s} c_s c_t \right) \left( \frac{p_s}{p_{n+1}}
+\frac{p_{n+1}}{p_s}+2\right)
+\frac{p_s}{p_{n+1}} \sum_{u \ne s, t\ne s} c_u c_t
+\frac{p_s}{p_{n+1}} \sum_{t\ne s} c_s c_t
-\frac{p_{n+1}}{p_s}\left(\sum_{s>t} c_s c_t +\sum_{s<t} c_s c_t\right)\nonumber \\
&=\left[c_s\left(\sqrt{\frac{p_s}{p_{n+1}}}+\sqrt{\frac{p_{n+1}}{p_{s}}}\right)+\sum_{t \ne s} \sqrt{\frac{p_s}{p_{n+1}}}c_t \right]^2.
\end{align}
\end{strip}


\section*{Acknowledgment}
The work of Q. Zhuang was supported by Defense Advanced Research Projects Agency (DARPA) under Young Faculty Award (YFA) Grant No. N660012014029 and Craig M. Berge Dean’s Faculty Fellowship of University of Arizona. 
C. Zhou would like to thank W. Evans and J. Guo for helpful consultations. The authors would like to thank Y. Polyanskiy for bringing Ref.~\cite{ordentlich2020strong} to our attention and pointing out its connection to our Theorem \ref{theo:1}.
\ifCLASSOPTIONcaptionsoff
  \newpage
\fi




\bibliographystyle{IEEEtran}
\bibliography{IEEEabrv,noisyNN.bib}
\end{document}